\documentclass{article}
\usepackage{tikz}
\usepackage{subcaption}
\usepackage{amsfonts}
\usepackage{amssymb}
\usepackage{graphicx}
\usepackage{amsmath}
\usepackage{amsthm}
\usepackage{pifont}
\usepackage{xcolor}
\usepackage{enumitem}
\usepackage{mathtools}
\usepackage[normalem]{ulem}
\usetikzlibrary{arrows.meta, bending}

\usepackage{tikz}
\usepackage{pgfplots}
\pgfplotsset{compat=1.18}

\usepackage[colorlinks=true,linkcolor=blue,citecolor= red,bookmarksopen=true,urlcolor=blue]{hyperref}
\nonstopmode
\newtheorem{theorem}{Theorem}[section]
\newtheorem{proposition}[theorem]{Proposition}
\newtheorem{lemma}[theorem]{Lemma}
\newtheorem*{lemma*}{Lemma}

\newtheorem{corollary}[theorem]{Corollary}

\newtheorem{definition}[theorem]{Definition}
\newtheorem{notation}[theorem]{Notation}
\newtheorem*{standingassumptionI*}{Standing Assumption I}
\newtheorem*{standingassumptionII*}{Standing Assumption II}
\newtheorem{assumption}[theorem]{Assumption}

\theoremstyle{remark}
\newtheorem{example}[theorem]{Example}
\newtheorem{remark}[theorem]{Remark}

\newcommand{\R}{{\mathbb R}}

\newcommand{\norm}[1]{\left\lVert#1\right\rVert}
\newcommand\abs[1]{\left|#1\right|}

\newcommand{\rn}[2]{\frac{\mathrm{d}#1}{\mathrm{d}#2}}

\newcommand{\fit}{{\mathcal{F}^{i}_t}}

\newcommand{\fib}{{\mathbb{F}^{i}}}
\newcommand{\hi}{{\mathcal{H}^{i}}}
\newcommand{\ki}{{K_{i}}}

\newcommand{\mi}{{M^{i}}}
\newcommand{\mib}{{M^{i,\infty}}}
\newcommand{\mie}{M^i_{e}}
\newcommand{\mibe}{M^{i,\infty}_{e}}
\newcommand{\pii}{\rho_{i,+}}

\renewcommand{\emph}[1]{\textit{#1}}
\newcommand{\red}[1]{\textcolor{red}{#1}}

\newcommand{\blue}[1]{{\color{blue}#1}}

\newcommand{\mcY}{\mathcal {Y}}
\newcommand{\mcF}{\mathcal {F}}

\newcommand{\mcR}{\rho^{\mathcal Y}_{+}}

\newcommand{\co}{\mathrm{cone}}
\newcommand{\hY}{Y^{\perp}}

\newcommand{\Me}{\mathcal M_e(\mcY)}

\usepackage{graphicx} % Required for inserting images

\title{Collective completeness and pricing-hedging duality II}
\author{
Alessandro Doldi\thanks{Dipartimento di Matematica, Università degli Studi di Milano, Via Saldini 50, 20133 Milano, Italy,
\emph{alessandro.doldi@unimi.it}.}, $\quad$ Marco Frittelli\thanks{Dipartimento di Matematica, Università degli Studi di Milano, Via Saldini 50, 20133 Milano, Italy,
\emph{marco.frittelli@unimi.it}.  }, $\quad$
Marco Maggis\thanks{Dipartimento di Matematica, Università degli Studi di Milano, Via Saldini 50, 20133 Milano, Italy,
\emph{marco.maggis@unimi.it}.}
}

\date{\today}

\begin{document}
\maketitle

\begin{abstract}
\noindent  
This paper complements and extends Doldi, Frittelli and Maggis, \emph{Collective completeness and pricing–hedging duality}, Math. Finan. Econ. 19, 757–784 (2025), by studying collective pricing and hedging when admissible risk exchanges form a finitely generated convex cone. The collective First Fundamental Theorem of Asset Pricing and the collective pricing–hedging duality are extended to this setting. A key contribution is a closedness result showing that no collective arbitrage implies the closedness of the aggregate feasibility cone combining infinite-dimensional trading opportunities with finite-dimensional exchanges.  The paper also proves that no-collective-arbitrage prices for vectors of contingent claims form a relatively open convex set. Finally, strong collective replicability is introduced and shown to be equivalent to price uniqueness. This leads to an enhanced collective Second Fundamental Theorem of Asset Pricing, providing equivalent characterizations of collective completeness and strong collective completeness in terms of the uniqueness of the collective equivalent martingale measure.
We highlight that  several core aspects of the theory are substantially altered when exchanges belong to a convex cone rather than a vector space.
\end{abstract}



%%%%%%%%%%%%%%%%%%%%%%%%%%

\section{Introduction}

Many modern financial environments involve several decision makers who can trade in (possibly) different markets and under different information structures, yet are also able to share risk through explicit transfers or contractual arrangements. Market segmentation, heterogeneous trading constraints, and asymmetric information may prevent any individual agent from achieving efficient risk allocation on their own. At the same time, cooperation through \emph{zero--sum} exchanges can expand the set of feasible payoffs available to the group. The theory of \emph{collective finance} formalizes this interaction by embedding individual trading opportunities within a multi-agent system in which agents may reallocate wealth among themselves subject to an aggregate budget constraint.

\medskip

Traditional asset pricing theory is primarily formulated in terms of a single representative agent operating in a frictionless market. Following the seminal contributions of Kreps \cite{Kreps81}, Harrison--Kreps \cite{HarrisonKreps79}, and Harrison--Pliska \cite{HarrisonPliska81}, market viability and asset valuation have been characterized through the Fundamental Theorem of Asset Pricing (FTAP), which establishes a deep connection between the absence of arbitrage opportunities and the existence of suitable martingale measures.

\medskip

Expanding upon the framework established in \cite{BDFFM25}, \cite{DFM25}, and \cite{F25}, we further develop the theory of collective finance along three primary directions. First, we extend the First FTAP (Theorem \ref{IFTAP:cone}) and the Pricing Hedging Duality (Theorem \ref{duality}) to a larger class of admissible exchanges. Second, we investigate the properties of the set of No-Collective-Arbitrage prices (Proposition \ref{propRO}). Third, we introduce the concept of strong collective replicability, establishing its equivalence with price uniqueness (Proposition \ref{propUnique}). These results deliver an enhanced version of the collective Second Fundamental Theorem of Asset Pricing (Theorem \ref{completeTHbis}), which formally links the notions of completeness and strong completeness to the uniqueness of the collective equivalent martingale measure.

\medskip

Before turning to the formal setting, it is useful to situate this viewpoint within a broader literature on cooperation in multi-agent systems.
 While this body of work is not directly comparable to the present setting, it rests on a similar idea: cooperation can be consistent with individual rationality and may enable agents to reach outcomes that would be inaccessible separately. In the present framework, mutually advantageous exchanges are characterized by compatibility between agents’ preferences and collective pricing measures  \cite{DFM26c}; in multi-agent systems, cooperation is typically understood as coordinated behavior directed toward common or mutually reinforcing objectives. Relevant contributions include \cite{DORAN_FRANKLIN_JENNINGS_NORMAN_1997}, shared intentionality \cite{Bratman1992}, commitments and conventions \cite{Jennings1993}, coordination through partial global planning \cite{DurfeeLesser1988,DeckerLesser1995}, and spontaneous cooperation via self-organization \cite{Steels1990}.
\medskip

Returning to the financial model, we work in discrete time and consider $N\ge 1$ agents. Agent $i$ trades using a personal filtration $\mathbb{F}^i$ and a set of discounted assets
available to her; the associated terminal trading gains form a  vector space $K_i$. Classical viability for agent $i$
is encoded by the individual no--arbitrage condition $\mathbf{NA}_i$. 
In addition to trading in their respective financial markets, agents may cooperate through \emph{risk exchanges} represented by vectors
\[
Y=(Y^1,\dots,Y^N)
\]
of random variables. The component \(Y^i\) denotes the net payoff transferred to agent \(i\) through the exchange mechanism, so that positive realizations correspond to capital inflows and negative realizations to capital outflows. These exchanges provide a mechanism for reallocating risks across agents, thereby creating opportunities that may not be attainable through individual market trading alone.
In this paper, the set $\mcY$ of
\emph{allowed} exchanges  is assumed to be a finitely generated convex cone
containing deterministic zero-sum transfers
\[
\R^N_0 := \left\{x\in\R^N : \sum_{i=1}^N x_i = 0 \right\},
\]
see Example \ref{excone}. The notion of \emph{Collective Arbitrage} captures the possibility that coordinated actions create gains that are invisible at the individual level. A collective arbitrage occurs if there exist trading
gains $k=(k^1,\dots,k^N)\in K_1\times\cdots\times K_N$ and an exchange $Y\in \mcY$ such that $k^i+Y^i\ge 0$ $P$-a.s. for all $i$, and $k^j+Y^j>0$ with positive probability for at least one agent
$j$. The corresponding viability condition, \emph{No Collective Arbitrage} $\mathbf{NCA}(\mcY)$, first introduced in \cite{BDFFM25}, requires
that no such favorable outcome can be achieved. 
In contrast to the classical single-agent framework, the condition $\mathbf{NCA}(\mcY)$ may still be satisfied even if the aggregate market, obtained by pooling all assets, admits arbitrage opportunities. This highlights the impact of market segmentation and restrictions on exchanges.

\medskip

A central object in the theory is the set of \emph{collective equivalent martingale measures} $\mathcal{M}_e(\mcY)$. A vector $\mathbf{Q}=(Q^1,\dots,Q^N)$ of probability measures is called a collective equivalent martingale measure if each $Q^i$ is equivalent to $P$, prices correctly the assets available to agent $i$, and $\mathbf{Q}$
satisfies the aggregate exchange condition
\[
\sum_{i=1}^N E_{Q^i}[Y^i]\le 0
\qquad \text{for all } Y\in \mcY,
\]
together with suitable integrability requirements (see \eqref{MartingaleMeasures}). When $\mcY$ is a vector space, the above inequality reduces to equality, whereas for a cone it captures the asymmetric nature of admissible transfers. 

\medskip

The first key contribution (Proposition \ref{C:closed}) of the paper is a closedness result for the aggregate feasibility cone
\[
{\sf X}_{i=1}^{N} (K_i - L^0_+(\Omega,\mathcal{F}^i_T,P)) + \mcY,
\]
which combines infinite-dimensional trading opportunities with finite-dimensional exchanges. We show that $\mathbf{NCA}(\mcY)$ already guarantees the required closedness, leading to the extension
of the collective First Fundamental Theorem of Asset Pricing. This result supports a collective pricing-hedging duality  for possibly unbounded claim vectors and highlights the role of cone-valued exchanges in breaking
the symmetry between super\slash subhedging, see Theorem  \ref{duality} and Remark \ref{remMinus}.

\medskip

\noindent Apart from the mathematical challenge that working with a cone $\mcY$ poses, this choice is economically relevant, as it models more realistic cooperation: If $\mcY$ is a convex cone, exchanges are directional: agents can agree on transfers that redistribute risk in allowed directions, but they cannot force arbitrary repayments in every direction. 
If instead $\mcY$ were a vector space, then for every admissible exchange $Y$ the opposite exchange 
$-Y$ would also be admissible. Economically, this means any transfer can be reversed at no cost and agents could neutralize any constraint by flipping the sign of exchanges.
\\ The cone structure prevents this drawback and ensures that exchanges represent credible and enforceable cooperation agreements (e.g. insurance-like risk sharing). This asymmetry is what allows collective arbitrage, pricing bounds, and fairness results to be economically more meaningful in the paper’s framework.
From a technical viewpoint, we emphasize that the proofs of Proposition \ref{C:closed}, Theorem \ref{IFTAP:cone} and Theorem \ref{duality} follow a substantially different line of argument from those developed in the corresponding results in the vector space case for $\mcY$ (respectively, Theorem 2.8, Theorem 2.1 and Theorem 2.7, \cite{DFM25}),  although they still rely on methods that are classical within this area of research.

\medskip

The second part of the paper focuses on no-collective-arbitrage prices. Given a vector of contingent claims $f=(f^1,\dots,f^N)$, we consider the no-collective-arbitrage prices set $\Pi(f)\subseteq\R^N$, as introduced in \cite{DFM25}. $\Pi(f)$ is the collection of those price vectors for which the enlarged market obtained
by introducing $f^i$ as a traded asset for agent $i$ still satisfies $\mathbf{NCA}(\mcY)$. When $\mcY$ is a vector space, we prove in Section \ref{property:pi} that $\Pi(f)$ is a \textit{relatively open} convex set.  As explained in Remark \ref{comparison}, this extends the classical one-dimensional result to the multi-agent setting.

\medskip

One of the key differences between collective and classical asset-pricing theory is that collective replicability does not, in general, imply price uniqueness, even when the admissible exchanges form a vector space. More precisely, for \(N\geq 2\), a collectively replicable claim may still admit multiple arbitrage-free prices; see Example~\ref{ex1}. This contrasts with the familiar single-agent framework, where replicability is typically associated with a unique arbitrage-free valuation.

To overcome this issue, for a given convex cone of exchanges \(\mcY\), we introduce the restricted cone
$\widehat{\mcY}\subseteq\mcY$,
consisting of exchanges that are individually fair under every collective equivalent martingale measure, rather than merely fair in aggregate.  This restriction removes the ambiguity generated by admissible exchanges whose aggregate value is zero but whose individual components may have nonzero values under different collective martingale measures. By restricting attention to these fair exchanges, we define strong collective replicability with respect to \(\widehat{\mcY}\). Our main structural result, presented in Section \ref{sec52}, shows that \textit{strong collective replicability is equivalent to price uniqueness and to the coincidence of collective superhedging and subhedging prices}. This leads to an enhanced collective Second
Fundamental Theorem of Asset Pricing, characterizing collective completeness and strong
collective completeness in terms of uniqueness of the collective equivalent martingale measure.

\medskip

We emphasize that, in the more general framework in which \(\mathcal Y\) is a convex cone rather than a vector space, several key aspects of the theory undergo significant changes.
A particularly relevant instance is the collective Second FTAP. Indeed, natural conjectures suggested by the corresponding results for  vector spaces of exchanges in \cite{DFM25} would be that: (i) under \(\mathbf{NCA}(\mcY)\), collective replicability---that is, replicability in each agent's market up to admissible exchanges---of a vector of contingent claims is equivalent to the equality of its collective subhedging and superhedging prices; or that (ii) collective completeness of the markets holds if and only if the set of equivalent collective martingale measures is a singleton. Both claims are, in fact, false in general, unless the joint assumptions of \(\mathbf{NCA}(\mathcal Y)\) and \(\mathbf{NCA}(-\mathcal Y)\) hold (compare Proposition \ref{uguali} with Proposition \ref{diversi} and the examples in Section \ref{example1}).
The same strengthened conditions are also needed in the statement of the collective Second Fundamental Theorem of Asset Pricing \ref{completeTHbis}. Even under these additional assumptions, the proofs of the aforementioned results require non-negligible adaptations relative to their counterparts in the vector space setting.

\section{Collective arbitrage in discrete time}\label{secsetting}

Let $\mathcal T=\{0, \dots, T\}$ be the finite set of discrete times and consider a given filtered probability space $(\Omega, \mathcal{F}, \mathbb F ,P)$, with  $\mathbb F =\{\mathcal{F}_t\}_{t \in  \mathcal T}$,  
%\sout{$\mathcal{F}_0=\{\emptyset, \Omega\}$} 
and $\mathcal{F}=\mathcal{F}_T$. We also assume that $\mcF_0$ is trivial, namely that $P(A)=0$ or $P(A)=1$ for all $A \in \mcF_0$. 
 We say that a probability measure $Q$ on $(\Omega, \mathcal{F})$ belongs to $\mathcal{P}_e$ if $Q \sim P$, or to $\mathcal{P}_{ac}$ if $Q\ll P$, respectively.
 Unless otherwise stated, all inequalities between random variables are meant to hold $P$-a.s.. \\

  The (global) securities market comprises a zero-interest rate riskless asset $X^0$ and $J$ risky assets with discounted price processes $X^j=(X^j_t)_{t\in \mathcal T}$, $j=1, \dots ,J$, $J\geq 1$. We set $X^0_t=1$ for all $t \in \mathcal T$.

Consider $N$ agents operating within the market, and assume that each agent $i$, $i=1,\dots,N$, can invest in the riskless asset and in the risky assets $X^j$, $j\in (i)$, where $(i)$ denotes a specified subset of $\{1,\dots,J\}$. We denote the cardinality of $(i)$ by $d_i$, representing the number of risky assets in which agent $i$ may invest, and assume $1 \leq d_i \leq J$.
We assume (without loss of generality) that $\cup_{i=1}^N (i)=\{1,\dots,J\}$, as
we may ignore the assets that cannot be used by any agent. We are not excluding that different agents may invest in the same risky assets nor that one or more agents may invest in the full market. Let $X^{(i)}:=(X^j,\, j \in (i))$ be the risky assets available for trading to agent $i$.
We denote by $$\fib=(\fit)_{t\in\mathcal{T}}\subseteq \mathbb F$$ the filtration representing the information available to agent $i$. Thus for all $i$, $\mathcal F^i_0=\mathcal F_0$, the trivial sigma-algebra.
We assume that all processes $X^j$, with $ j \in (i)$, are adapted with respect to the  filtration $\fib$.

For a detailed and exhaustive motivation behind the segmented market setup, the reader may consult \cite{BDFFM25}, \cite{Carassus23}, \cite{DFM25}  and  \cite{DFM26c}, Remark 4.1.

\begin{notation}
\label{productnotation}
If $\mathbf Q=(Q^1,\dots,Q^N) \in (\mathcal P_{ac})^{N} $ and $p \in [0,\infty]$ we set $$L^{p }(\Omega, \mathbf{F}_t,\mathbf Q):=L^{p }(\Omega, \mathcal{F}_t^{1},Q^1) \times \dots \times  L^{p }(\Omega, \mathcal{F}_t^{N},Q^N), \quad t\in \mathcal T .$$ 
With a slight abuse of notation, when we consider a single probability measure $Q \in \mathcal P$, we still use the same notation $$L^{p }(\Omega, \mathbf{F}_t,Q):=L^{p }(\Omega, \mathcal{F}_t^{1},Q) \times \dots \times  L^{p }(\Omega, \mathcal{F}_t^{N},Q),\quad t\in \mathcal T .$$
\end{notation}

 We say that a set $K \subseteq L^{0 }(\Omega, \mathbf{F}_T,P)$ is closed in $L^{0 }(\Omega, \mathbf{F}_T,P)$ if it is closed for the convergence in probability. 
 
 We adopt the following partial order among random vectors: for every $f,g\in L^{0}(\Omega, \mathbf{F}_T, P)$ we write $f \leq g$ if and only if $P(f^i\leq g^i)=1$ for every $i=1,\ldots, N$. We also let $L_+^{p }(\Omega, \mathbf{F}_T,P):=\{f \in L^{p }(\Omega, \mathbf{F}_T,P) \mid f \geq 0 \}$.
 %\blue{For any $a\in\R$ we shall always indicate with $\mathbf{a}$ the vector $\mathbf{a}=(a,a,\ldots,a)\in\R^N$.} 
 \\

A stochastic process $H=(H_t)_{t\in \mathcal T}$ is called an \label{admisstrat}\emph{admissible trading strategy for the agent} $i$ if it is $d_i$-dimensional and  predictable with respect to $\fib$.
The space of admissible trading strategies for the agent $i$ is denoted by $\hi$. 
If $H \in \hi$, we set 
\begin{equation*}
(H\cdot X^{(i)})_t:=\sum_{h \in (i) } (H^h\cdot X^h)_t,  \quad t \in \mathcal T,  
\end{equation*} 
where $(H^h\cdot X^h)_t:=\sum_{s=1}^t H^h_s(X^h_s-X^h_{s-1})$ denotes the stochastic integral of $H^h$ with respect to  the asset $X^h$, $h\in (i)$, and we write

\begin{equation}
\label{def:KiBis}
\ki:=\{(H\cdot X^{(i)})_T \mid H \in {\hi} \} \subseteq L^{0 }(\Omega, \mathcal{F}_T^{i},P).
\end{equation}

\noindent The sets of martingale measures for  the assets in $(i)$ are defined by
\begin{equation*}
\mi:=\left\{ Q \in \mathcal{P}_{ac} \mid     X^j  \text { is a } (Q, \fib) \text {-martingale for all } j \in (i) \right\},
\end{equation*}
\begin{equation}
\label{mifty}
M^{i,\infty}(P):=\left\{ Q \in M^i \mid  \frac{dQ}{dP} \in L^{\infty}(\Omega, \mathcal{F}_T^{i},P)  \right\},
\end{equation}
Moreover, we set $\mie:=\mathcal P_e \cap \mi $ and $M_e^{i,\infty}(P):=\mathcal P_e \cap M^{i,\infty}(P)$.
%\text {and} \quad \mibe(P):=\mathcal P_e \cap \mib(P). $$
The classical No Arbitrage condition for agent $i$ holds if:

\begin{equation*}
\label{eq:NAi}
 \mathbf{NA}_{i}: \text{  } \ki \cap L_{+}^{0}(\Omega, \mathcal{F}^i_T , P)=\{0\}.
 \end{equation*}
While the space $K_i$ is closed in probability by Stricker's Theorem, only under the $\mathbf{NA}_{i}$ condition also the set $(\ki - L_{+}^{0}(\Omega, \mathcal{F}^i_T , P))$ is closed in $L^{0}(\Omega, \mathcal{F}^i_T , P)$ (see \cite{DS2006} Theorem 6.9.2), which is an essential property for the proofs of both the first FTAP and the superhedging duality.
\begin{theorem}[Dalang, Morton and Willinger (1990) \cite{DMW90}]\label{DMW}
In the discrete time setting described above fix any $i$. If $X^{(i)}$ is integrable under $P$   then  
\begin{equation*}
 \mathbf{NA}_{i} \;\Longleftrightarrow\;  \mie \not = \emptyset \;\Longleftrightarrow\;M_e^{i,\infty}(P) \not = \emptyset.
 %\quad \mathbf{NA} \Longleftrightarrow  M \cap \mathcal{P}_e \not = \emptyset.
\end{equation*}
\end{theorem}

It is well-known that the integrability of $X^{(i)}$
  can be ensured by a change to an equivalent probability measure.

Finally, we recall the following standard fact: if \(Q \in M^i\) and \(k \in K_i\) is such that
$
k^- \in L^1(\Omega,\mathcal F_T^i,Q),
$
then it follows that
$
k \in L^1(\Omega,\mathcal F_T^i,Q)
\quad\text{and}\quad
E_Q[k]=0
$
(see \cite[Theorem 5.14]{FollmerSchied2}). We shall use this result repeatedly without further mention.

\subsection{Setting for collective arbitrage}

One of the key principles introduced in \cite{BDFFM25} is that, besides trading in their own individual markets, $N$ agents may improve their portfolios by engaging in cooperative risk transactions.
Let $\mathcal{Y}_0$ denote the set of all zero-sum risk exchanges, defined as
\begin{equation*} \label{2345}
\mathcal{Y}_0 = \left\{ Y \in L^0(\Omega, \mathbf{F}_T, P) \mid \sum_{i=1}^N Y^i = 0  \right\}.
\end{equation*}

For every $Y \in \mathcal{Y}_0$, the aggregate sum of its components is $P$-almost surely equal to zero, whereas the single entries $Y^i$ are generally random variables. A positive value of $Y^i$ on a given event represents a transfer of capital from the group to agent $i$, while a negative value corresponds to a transfer from agent $i$ to the others. Therefore, each element $Y \in \mathcal{Y}_0$ describes a possible redistribution of wealth among the agents under the requirement of zero total net transfer. 
An example of permissible exchanges is a convex cone $\mathcal{Y} \subseteq \mathcal{Y}_0$ (see Assumption \ref{ass:cone} for the exact class we shall consider).

When agent $i$ follows an investment strategy in its own market $(i)$, she will obtain a terminal payoff $k^i \in {K}_i$.
The agents may also enter in the risk exchange corresponding to a vector $Y \in \mathcal Y $. This procedure leads to the terminal time value $k^i + Y^i$ for agent $i$.
In \cite{BDFFM25} a \emph{Collective Arbitrage} consists of vectors $(k^1, \dots , k^N) \in {K}_1 \times \cdots \times {K}_N$ and $Y =(Y^1, \dots, Y^N) \in \mcY$ satisfying
\begin{align*}
 k^i+Y^i\geq 0  \quad P\text{-a.s. } & \quad \forall \, i\in\{1,\dots,N\} \, \text{ and } 
\\ P(k^n+Y^n>0)>0   &\quad \text{for at least one } n \in \{1,\dots,N\}.
\end{align*}

 \begin{definition}[No Collective Arbitrage, Def 3.1 \cite{BDFFM25}]
\label{NCA} 
No Collective Arbitrage for $\mathcal{Y}$ ($\mathbf{NCA}(\mathcal{Y})$) holds if
\begin{equation*}
({\sf X}_{i=1}^{N} K_i+\mathcal Y )\cap  L^{0 }_+(\Omega, \mathbf{F}_T,P)=\{0\}, \label{NCAY} 
 \end{equation*}
 where ${\sf X}_{i=1} ^{N} \ki$ denotes the Cartesian product of the sets $\ki$ defined in \eqref{def:KiBis}.
 \end{definition}
 
 As shown in Proposition 3.2  \cite{BDFFM25},
 \begin{equation}\label{NCACC}
      \mathbf{NCA}(\mathcal{Y}) \Leftrightarrow K^\mathcal Y \cap  L^{0 }_+(\Omega, \mathbf{F}_T , P)=\{0\} \Leftrightarrow C^\mathcal Y \cap  L^{1 }_+(\Omega, \mathbf{F}_T , P)=\{0\}, 
  \end{equation}where 
 \begin{equation}\label{KK}
      K^\mathcal Y:=  {\sf X}_{i=1} ^{N} ( K_i - L^{0 }_+(\Omega, \mathcal{F}^i_T,P) ) + \mathcal Y \text{ and } C^\mathcal Y := K^\mathcal Y \cap L^{1 }(\Omega, \mathbf{F}_T, P)
 \end{equation}

\begin{remark}
As explained in the aforementioned references where $\mathbf{NCA}(\mcY)$ was introduced, the concept of no collective arbitrage is a genuine novel notion that lies (strictly) between no classical arbitrage for the global market and no classical arbitrage for all individual agents.  In addition, (see Remark 4.5 \cite{DFM26c}) no collective arbitrage for $\mcY $ is not in general equivalent to classical
no–arbitrage for a suitably defined representative (social–planner) agent. 
\end{remark}

\begin{remark} The notion of arbitrage developed in the strand of literature initiated by
\cite{Page95} and further pursued, among others, in \cite{Page2006,DanaLeVan2010}
is a preference-based collective concept, formulated in a finite-state one-period setting. It does not rely on the existence of a
market of dynamically traded assets. Rather, an exchange is regarded as useful
when it improves agents' utilities asymptotically, a property that, in the
framework of \cite{DanaLeVan2010}, can be characterized in terms of
risk-adjusted sets of priors. In particular, the components of this type of arbitrage opportunity need
not be statewise nonnegative for each agent; it is sufficient that each agent
evaluates her own component as a useful direction according to her preferences,
ambiguity attitude, and risk-adjusted beliefs.

By contrast, the notion of collective arbitrage considered in the present paper
is market-based. Each agent trades in her own financial submarket and may also
participate in a collective exchange mechanism. A collective arbitrage occurs
when the resulting final payoff of every agent, after individual trading and
collective exchange, is statewise nonnegative, and strictly positive for at least one agent on an event of positive probability. Preferences play no explicit role in this definition. Accordingly, the associated duality is formulated in terms of
vectors of martingale measures satisfying a compatibility condition induced by
the admissible exchange set, rather than in terms of utility gradients or
risk-adjusted priors.

\end{remark}

\subsection{The collective FTAP for convex cones $\mcY$}

In this section we propose an extension of  \cite{DFM25} Theorems  2.1 and  2.7,  by relaxing the assumption on $\mcY$ being a vector space. The idea of adopting directional exchanges was already inspected in the original contribution \cite{BDFFM25}, even though those results were obtained under the additional assumption of the closure in probability of $K^{\mcY}$. Here we show that this assumption is redundant if $\mcY$ is finitely generated.
To this aim we briefly recall the main definition and statements.

\begin{definition}[\cite{DFM25} Definition 1.9]
The sets of equivalent collective martingale measures are defined by
\begin{equation}\label{MartingaleMeasures}
\begin{split}
    \Me=\bigg\{ \mathbf Q=(Q^1,\dots,Q^N) \in {\sf X}_{i=1} ^{N}\mie   \mid  \mcY \subseteq L^{1}(\Omega, \mathbf{F}_T,\mathbf Q)  \text { and  } \sum_{i=1}^N E_{Q^i}[Y^i]\leq 0 \;\forall\, Y \in \mcY  \bigg\}
\end{split}
\end{equation}
\end{definition}

 \begin{assumption}
 \label{ass:cone} The set of allowed exchanges $\mcY $ is a finitely generated convex cone satisfying $\R^N_0 \subseteq \mcY \subseteq L^{0 }(\Omega, \mathbf{F}_T,P)$,  where 
 \begin{equation*}\label{R0}
\mathbb R_0 ^N:=\left \{   x \in \mathbb R^N \mid \sum_{i=1} ^N x^i =0  \right \}=\mathrm{span}\left \{ e_i-e_j \mid i,j \in \{1,\dots,N\} \right \}
\end{equation*}
and $\{e_i\}_{i=1,\dots,N}$ is the canonical basis in $\R^N$.
 \end{assumption}

Notice that a finitely generated convex cone can be assumed to satisfy\footnote{This can be easily obtained by standard arguments via an equivalent change of measure $\widehat P\sim P$} $\mcY \subseteq L^{1}(\Omega, \mathbf{F}_T,P)$. It can be written as 
\begin{equation}\label{coY}
    \mcY:=\mathrm{cone}(\{Y_1, \dots , Y_R\} )=\bigg \{ \sum_{m=1}^R \alpha_m Y_m  \mid \alpha_m \in \R, \,\, \alpha_m\ \geq 0 \text{ for all } m\bigg \},
\end{equation}
%for $Y^0:=\mathbb R^N_0 $ and
namely the cone generated by a finite number $R \in \mathbb N$ of vectors $Y_m \in L^{1 }(\Omega, \mathbf{F}_T,P), $ $m=1,\dots,R $. When we want to ensure that $\R^N_0 \subseteq \mcY $, we assume that the vectors  $(e^i-e^j)_{i,j}$  belong to the collection $\{Y_1,\dots,Y_R\}$ generating $\mcY$.

\begin{example}\label{excone}

The following examples illustrate admissible exchange sets that are finitely generated convex cones rather than vector spaces.
Fix \(t\in\{0,\dots,T\}\).

\begin{enumerate}
\item Let
$\mcY=\operatorname{cone}(\{Y_1,\dots,Y_R\})$,
where each $Y_m\in L^0(\Omega,\mathbf F_t,P)$ satisfies $\sum_{i=1}^N Y_m^i \leq 0,
\, P\text{-a.s.}$,
for \(m=1,\dots,R\). In this case, every admissible exchange can only redistribute wealth among agents while generating a non-positive aggregate transfer.

\item Let $\mcY=\operatorname{cone}(\{Y_1,\dots,Y_R\})$,
where each \(Y_m\in L^0(\Omega,\mathbf F_t,P)\) satisfies
$F(Y_m)\leq 0,
\, m=1,\dots,R,
$
for some sublinear value functional
\[
F:L^0(\Omega,\mathbf F_t,P)\to\mathbb R,
\]
for which $F(x) \leq 0$ if  $x \in \mathbb R_0 ^N$. Here, sublinearity means that
$F(\alpha Y+\beta Z)
\leq
\alpha F(Y)+\beta F(Z),
\,
\alpha,\beta\geq0,$
 \(Y,Z\in L^0(\Omega,\mathbf F_t,P)\).
This framework includes, for instance, $F(Y):=\rho(-Y)$, for systemic coherent risk measures $\rho$, for which numerous examples can be found in the literature, for instance in \cite{BFFMB}. We stress that the use of $-Y$ inside the risk measure functional is motivated by the fact that here positive values stand for profits, while in the risk measures' literature positive values stand for losses. A particularly simple specification is
$
F(Y):=\sum_{i=1}^N \rho^i(-Y^i),
$
where $\rho^1,\dots,\rho^N$ are coherent risk measures (e.g. $\rho^i(Y^i)=E_{P^i}[-Y^i]$ with $P^i\ll P$).
\end{enumerate}
\end{example}

The main mathematical challenge in the results of Section \ref{secsetting} is to show that under $\mathbf{NCA(\mathcal{Y})}$ the cone $K^{\mcY}$ defined in \eqref{KK} is actually closed in probability. The proof  is non-trivial as $K^{\mcY}$ is the sum of an infinite dimensional closed cone and a finite dimensional one.
Indeed, the no-arbitrage condition plays a crucial role to obtain the closure property of $K^{\mcY}$, as can be easily understood by the following example, which shows that the closure of the sum of two closed convex cones easily fails already in finite dimensional spaces.   

\begin{example}
Consider $\mathbb{R}^4$ with its usual topology. Define
\[
K := \Bigl\{(w,x,y,z)\in\mathbb{R}^4 \mid
w\ge 0,\ x\ge 0,\ y^2 + (z-x)^2 \le x^2 \Bigr\}.
\]
It is easy to verify that $K$ is a closed convex cone. Similarly for
\[
C := \{\, t(-1,-1,0,0) \mid t\ge 0 \,\}.
\]
We show that the convex cone $K+C$ is not closed. For $n\in\mathbb{N}$, set
\[
x_n := \frac{n^2+1}{2n}, \qquad
k_n := (x_n,x_n,1,1/n), \qquad
c_n := x_n(-1,-1,0,0).
\]
Then $k_n\in K$, $c_n\in C$, and
\(
k_n + c_n = (0,0,1,1/n) =: s_n \in K + C.
\)
Moreover,
\(
s_n \longrightarrow s := (0,0,1,0) \,\,\text{in } \mathbb{R}^4\), so that $s$ belongs to the closure of $K+C$.

If $s\in K+C$, then $s=k+c$ for some $k\in K$ and $c=t(-1,-1,0,0)\in C$
with $t\ge 0$. Hence \(k = (t,t,1,0).
\)
But $k\in K$ would imply
\(
1^2 + (0-t)^2 \le t^2,
\)
which is impossible. Therefore $s\notin K+C$.
\end{example}

\begin{proposition}\label{C:closed}
Suppose that the set of allowed exchanges $ \mcY \subseteq L^{0 }(\Omega, \mathbf{F}_T,P)$ is a finitely generated convex cone   and $\mathbf{NCA(\mathcal{Y})}$ holds true. Then 
$K^{\mcY}$ is closed in  $L^{0}(\Omega, \mathbf{F}_T,P)$.
\end{proposition}

\begin{proof}

 We recall from \cite{DFM25}, equation (4), that the implication $\mathbf{NCA}(\mcY) \Rightarrow \mathbf{NA}_i$ for every $i=1,\ldots,N$ holds trivially. The proof  is then an immediate consequence of the following general lemma, as in our financial setup, for each $i$, the vector space $K_i$ and the cone $K_i-L^0_+(\Omega, \mathcal{F}^i_T,P)$ are closed in probability, as a consequence of $\mathbf{NA}_i$.  Thus the vector space $ K:=({\sf X}_{i=1} ^{N}  K_i)\subseteq L^0(\Omega, \mathbf{F}_T,P)$ and the cone  $K- L^{0 }_+(\Omega, \mathbf{F}_T,P)=  {\sf X}_{i=1} ^{N} ( K_i - L^{0 }_+(\Omega, \mathcal{F}^i_T,P) ) \subseteq L^0(\Omega, \mathbf{F}_T,P)$  are closed in probability. Moreover, with these notations, $\mathbf{NCA}(\mcY)$ can be written as $(K+\mcY)\cap L^0_+(\Omega, \mathbf{F}_T,P)=\{0\} $ and
 $K^{\mcY}=K+\mcY- L^0_+(\Omega, \mathbf{F}_T,P)$.

 \end{proof}

\begin{lemma}
    Let
    \begin{enumerate}
        \item $K \subseteq L^0(\Omega, \mathbf{F}_T,P)$ be any vector space with both $K$ and $K-L^0_+(\Omega, \mathbf{F}_T,P)$ closed in probability.
        \item $\mcY =\mathrm{cone}(\{Y_1,\dots,Y_R\}) \subseteq L^0(\Omega, \mathbf{F}_T,P)$ be any finitely generated convex cone, with $Y_m \in L^0(\Omega, \mathbf{F}_T,P)$, $m=1,\dots,R$.
        \end{enumerate}

If $(K+\mcY)\cap L^0_+(\Omega, \mathbf{F}_T,P)=\{0\} $ then $K+\mcY- L^0_+(\Omega, \mathbf{F}_T,P)$ is closed in probability.
\end{lemma}

\begin{proof}
We use the notation $\norm{x}$ for the usual $2-$norm on $\R^N$ and define the probability measure $\tilde P \in \mathcal P_e$ by $\rn{\tilde{P}}{P}=\frac{c}{1+\sum_{m=1}^R \norm{Y_m}^2}$. Then $\mcY\subseteq  L^2(\Omega, \mathbf{F}_T,\tilde{P})$.
Since $K$ is closed in probability, then $L^2(\Omega, \mathbf{F}_T,\tilde{P})\cap K$ is closed in $L^2(\Omega, \mathbf{F}_T,\tilde{P})$. Let $\pi,\pi_\perp$ be the projections, defined on the Hilbert space $L^2(\Omega, \mathbf{F}_T,\tilde{P})$, onto the closed subspaces $L^2(\Omega, \mathbf{F}_T,\tilde{P})\cap K, \,(L^2(\Omega, \mathbf{F}_T,\tilde{P})\cap K)^\perp$ respectively.
    We first show that 
    \begin{equation}
    \label{useperp}
        K+\mcY=K+\mathrm{cone}\Big(\{\pi_\perp(Y_1),\dots, \pi_\perp(Y_R)\}\Big).
    \end{equation}
    Indeed, to prove $(\subseteq)$, it is easy to see that any $k+Y\in K+\mcY$ can be written as
    $k+\sum_{m=1}^R\alpha_m Y_m$ for $\alpha_m\geq 0,\, m=1,\dots,R$. Then it is enough to write  $Y_m=\pi(Y_m)+\pi_\perp(Y_m)$ and to note that $\pi(Y_m)$ belongs to $K$, which is a vector space. As to $(\supseteq)$, observe that for $\alpha_m\geq0, m=1,\dots,R$
    $$k+\sum_{m=1}^R\alpha_m\pi_\perp(Y_m)=k-\sum_{m=1}^R\alpha_m\pi(Y_m)+\sum_{m=1}^R\alpha_m\big(\pi(Y_m)+\pi_\perp(Y_m)\big)=k-\sum_{m=1}^R\alpha_m\pi(Y_m)+\sum_{m=1}^R\alpha_mY_m$$ and $k-\sum_{m=1}^R\alpha_m\pi(Y_m)\in K$, while the last term in the right hand side belongs to $\mcY$. 

\medskip
    
\noindent \textbf{Proof of the closure property.} We now prove the thesis in two steps and to this aim take a sequence 
    $k_n+Y_n -l_n$ in $K+\mcY- L^0_+(\Omega, \mathbf{F}_T,P)$ converging in probability to some random vector $f$. Passing to a subsequence and relabeling we take the convergence to be a.s. ($P$ and $\tilde{P}$ by equivalence):
\begin{equation}
    \label{convtofBis}
    k_n+Y_n-l_n\longrightarrow_n f\quad \text{a.s.}.
\end{equation}
As in the proof of \cite{Aliprantis} Corollary 5.25, without loss of generality we may and will assume the existence of a linearly independent subset of $\{\pi_\perp(Y_1),\dots, \pi_\perp(Y_R)$ \}.
If \(\pi^\perp(Y_m)=0\) for every \(m=1,\ldots,R\), and hence
\(Y_m=\pi(Y_m)\in K\), we have \(\mathcal Y\subseteq K\), so that
\[
K+\mathcal Y-L^0_+(\Omega,\mathbf F_T,P)
=
K-L^0_+(\Omega,\mathbf F_T,P),
\]
which is closed by assumption. Thus there is nothing to prove, and we may
assume without loss of generality that at least one of the vectors $\{\pi_\perp(Y_1),\dots, \pi_\perp(Y_R)\}$ is not null. In particular, the linearly independent subset of $\{\pi_\perp(Y_1),\dots, \pi_\perp(Y_R)\}$ will consist of at least one non-zero element.
In the remainder of the proof we denote such a subset of linearly independent elements by
$$\{\hY_1, \dots, \hY_M\},$$ 
with $\hY_m \in (L^2(\Omega,\mathbf{F}_T,\tilde{P})\cap K)^\perp$, $m=1,\dots,M \leq R$, and a subsequence, which we shall denote by $Y_n$ again, such that 
\begin{equation}
    \label{YY}
   Y_n=\sum_{m=1}^M\alpha_m^n \hY_m,
\end{equation}
with all coefficients $\alpha_m^n \geq 0$. Set $\beta_n=\sum_{m=1}^M\alpha_m^n \geq 0$ and note that
$\sum_{m=1}^M \frac{\alpha_m^n}{\beta_n} =1 $ for each $n$. 

\medskip

\noindent\textbf{Step 1:} We first show $\limsup{\beta_n}<\infty$.  Suppose by contradiction that $\limsup{\beta_n}=\infty$. Then by passing to a subsequence we may always assume that $\lim \beta_n=\infty$ and so from \eqref{convtofBis} and \eqref{YY}
\begin{equation}
    \label{convtof}
    \frac{k_n}{\beta_n}-\frac{l_n}{\beta_n}+\frac{\sum_{m=1}^M\alpha_m^n \hY_m}{\beta_n}\longrightarrow_n 0\quad \text{a.s.}.
\end{equation}
Since $0\leq \frac{\alpha_m^n}{\beta_n} \leq 1 $ for all $m=1,\dots,M$, the vector $(\frac{\alpha_1^n}{\beta_n},\ldots,\frac{\alpha_M^n}{\beta_n})$ belongs to the compact set of $\R^M$ 
\[\left\{(a_1,\ldots,a_M)\in \R^M\mid a_m\geq 0 \;\forall m=1,\ldots,M \text{ and } \sum_{m=1}^M a_m=1\right\}. \]
We can therefore pick a further subsequence $h_n$ for which $\frac{\alpha_m^{h_n}}{\beta_{h_n}}$ is converging for all $m$ and their sum is still $1$. Thus by relabeling again we obtain 
$$\frac{\alpha_m^n}{\beta_n}\rightarrow_n \alpha_m^\infty\geq 0,\,\forall m=1,\dots,M, \,\,  \text { and } \sum_{m=1}^M\frac{\alpha_m^n}{\beta_n} \hY_m \rightarrow_n \sum_{m=1}^M \alpha_m^\infty \hY_m:=Y_\infty, \,\, a.s.$$
with $\sum_{m=1}^M \alpha_m^\infty=1$ and $Y_\infty\in (L^2(\Omega, \mathbf{F}_T,\tilde{P})\cap K)^\perp \cap \co\Big(\hY_1,\dots, \hY_M \Big)$.
Furthermore, since $\sum_{m=1}^M\frac{\alpha_m^n}{\beta_n} \hY_m$ is a.s. converging, we have from \eqref{convtof} that
 $\frac{1}{\beta_n}\Big(k_n-l_n\Big)$ is itself a.s. converging. Since such a sequence belongs to $K-L^0_+(\Omega, \mathbf{F}_T,P)$, closed in probability by hypotheses, we infer that its limit belongs to the latter cone and can be written as $k_\infty-l^+_\infty$ with obvious notation.
\\To sum up, for the limits we have
\begin{equation}
    \label{limits}
    k_\infty-l_\infty^++Y_\infty=0 \quad \Rightarrow \quad  k_\infty+Y_\infty=l_\infty^+  
\end{equation}
so that $k_\infty +Y_\infty \geq0$.
% \blue{and we can use $$(K+\mcY)\cap L^0_+(\Omega, \mathbf{F}_T,P)=\Big(K+\co(\{\hY_1,\dots, \hY_M\})\Big)\cap L^0_+(\Omega, \mathbf{F}_T,P)=\{0\}$$ (recall \eqref{useperp}) to conclude that $k_\infty+Y_\infty=0$. }
% \red{ credo che dovremmo evitare di far pensare, come potrebbe sembrare qui sopra, che $\mcY=\co(\{\hY_1,\dots, \hY_M\})$, perchè è falso (non tutti i finitely generated convex cone sono generati da vettori lin indipendenti).   quindi suggerisco di cancellare la parte sopra in blu e lasciare la seguente formulazione:}

From \eqref{useperp} we know that
$\Big(K+\co(\{\hY_1,\dots, \hY_M\})\Big) \subseteq \Big(K+\mathrm{cone}(\{\pi_\perp(Y_1),\dots, \pi_\perp(Y_R)\}\Big)=K+\mcY$. As $k_\infty \in K$ and $Y_\infty \in  \co\Big(\hY_1,\dots, \hY_M \Big)$, we can use $(K+\mcY)\cap L^0_+(\Omega, \mathbf{F}_T,P)=\{0\}$
to conclude that $k_\infty+Y_\infty=0$. 
This in turn implies that $Y_\infty\in K$ since the latter is a vector space. Thus 
$$Y_\infty\in \big(L^2(\Omega, \mathbf{F}_T,\tilde{P}) \cap K\big)\cap \big( L^2(\Omega, \mathbf{F}_T,\tilde{P}) \cap K\big)^\perp$$
i.e. $Y_\infty=0$. At the same time we have $$0=Y_\infty=\sum_{m=1}^M\alpha^\infty_m \hY_m\quad \text{ with } \alpha_m^{\infty} \geq 0 \, \, \forall m, \,\,\text{ and  }\sum_{m=1}^M\alpha_m^\infty=1$$
which contradicts the linear independence of $(\hY_1, \dots,\hY_M)$.

\medskip

\noindent \textbf{Step 2}: We have shown that $\beta:=\limsup{\beta_n}<+\infty$. We may pass to a subsequence and relabel in such a way that $0\leq \beta_n\rightarrow_n \beta $. By the triangular inequality
\[\norm{\sum_{m=1}^{M}\alpha^n_m\hY_m}_{L^2}\leq \sum_{m=1}^M\alpha^n_m\norm{\hY_m}_{L^2}\leq \max_{m=1,\ldots, M}\norm{\hY_m}_{L^2}\cdot \sum_{m=1}^M\alpha^n_m \rightarrow_n \beta \max_{m=1,\ldots, M}\norm{\hY_m}_{L^2}\]
so that the sequence $\sum_{m=1}^M\alpha^n_m\hY_m$ is norm bounded in the finite dimensional closed subspace $\mathrm{span}(\{\hY_m, m=1,\ldots,M\})$ and thus it admits a converging subsequence. Relabel again the indices ensuring that $\sum_{m=1}^M\alpha^n_m\hY_m\rightarrow_n Y$ where convergence takes place in $L^2(\Omega, \mathbf{F}_T,\tilde{P})$ and, up to a subsequence, also a.s.. Since $\sum_{m=1}^M\alpha^n_m\hY_m$ is a sequence in $\co(\hY_1,\dots, \hY_{M})$ and the latter is norm closed by \cite{Aliprantis} Corollary 5.25, we have $Y\in\co(\hY_1,\dots, \hY_{M})$. 
\\ From \eqref{convtofBis} and \eqref{YY} we see that $k_n-l_n\rightarrow_n f-Y $ a.s., and by the assumption of $K-L^0_+(\Omega, \mathbf{F}_T,P)$ being closed, we conclude that $f-Y=k-l$ for some $k\in K$ and $l\in L^0_+(\Omega, \mathbf{F}_T,P)$. 
Then we have 
$$f=k-l+Y\in K-L^0_+(\Omega, \mathbf{F}_T,P)+\co(\hY_1,\dots, \hY_{M})$$
and using \eqref{useperp} we get also $f\in K-L^0_+(\Omega, \mathbf{F}_T,P)+\mcY $ as desired.
\end{proof}

\begin{remark}
The condition $(K+\mcY)\cap L^0_+(\Omega,\mathbf{F}_T,P)=\{0\}$
is required in the proof of the closedness of the set
$K+\mcY-L^0_+(\Omega,\mathbf{F}_T,P)$.
However, the set $K+\mcY$ remains closed even in the absence of this assumption. Indeed, it suffices to repeat the previous argument after replacing $l_n$, $l_\infty$, and $l$ by $0$, and substituting $K-L^0_+(\Omega,\mathbf{F}_T,P)$ with $K$ throughout.
Under these modifications, the only potentially delicate step concerns \eqref{limits}. In the present case, however, that identity simply reduces to
$k_\infty+Y_\infty=0$,
so that no additional assumption is needed. All remaining steps then follow verbatim from the preceding proof after the above changes.
\end{remark}

\begin{remark} \label{remPP}
    Consider a finitely generated convex cone $$\mcY=\mathrm{cone}(\{ Y_1,\dots,Y_R \}),$$ containing $\mathbb R^N_0$, with $ Y_m \in L^{0}(\Omega, \mathbf{F}_T,P)$, $m=1,\dots,R$ and let $f \in L^{0}(\Omega, \mathbf{F}_T,P)$. Fix any   $\varphi \in L^{0}(\Omega, \mathbf{F}_T,P)$ such that
$|f|\leq \varphi$.
    By setting 
    \begin{equation}
        \label{changemeas}
      \frac{d \widehat P}{dP}:= \frac{c}{1+\sum_{j,t}| X^j_t |+\sum_{m,i}|Y^i_m|+\sum_{i}|\varphi^i|}  
    \end{equation}
     for some positive normalizing constant $c$, we see that $\widehat P \in \mathcal P_e$, $\frac {d\widehat P} {dP} \in L^{\infty }(\Omega, \mathcal{F},P)$, all processes $X^1,\dots,X^J$ are integrable under $ \widehat P$, each element $ Y_m \in L^{1 }(\Omega, \mathbf{F}_T, \widehat P)$,  $\varphi \in L^{1 }(\Omega, \mathbf{F}_T, \widehat P)$ and $f \in L^{1 }(\Omega, \mathbf{F}_T, \widehat P)$. In addition, $\mcY \subseteq L^{1 }(\Omega, \mathbf{F}_T, \widehat P) \subseteq L^{1 }(\Omega, \mathbf{F}_T, \mathbf{Q}) $ and $f \in L^{1 }(\Omega, \mathbf{F}_T, \mathbf{Q})$ for each  $\mathbf{Q}$ with $Q^i \in \mathcal P_{ac}$ and $ \rn{Q^i}{ \widehat P} \in L^{\infty}(\Omega, \mathcal{F}^i_T, \widehat P) $, $i=1,\dots,N$.
\end{remark}

We state the First FTAP in the collective setting when $\mcY$ is a cone.

\begin{theorem}\label{IFTAP:cone}
    Under  Assumption \ref{ass:cone} 
    \begin{equation*}
      \mathbf{NCA(\mathcal{Y})} 
      %\iff \mathcal M_e^{\infty}(\mcY)\not = \emptyset 
      \iff \Me \not = \emptyset.  
    \end{equation*}
    Furthermore, if $\mathbf{NCA}(\mcY)$ holds, for every vector $\mathbf{P}=(P^1,\dots, P^N)$ of probability measures with $P^i\sim P$, we have
    \begin{equation}\label{MMM}
\mathcal M_{e}^{\infty}(\mcY,\mathbf{P}):=\left\{ \mathbf Q \in {\sf X}_{i=1} ^{N}\mibe( P^i)   \mid \mcY \subseteq L^{1}(\Omega, \mathbf{F}_T,\mathbf Q), \sum_{i=1}^N E_{Q^i}[Y^i]\leq 0 \text {  }\forall Y \in \mcY  \right\} \neq \emptyset
\end{equation}
for $\mibe(P^i)$ given in \eqref{mifty}.
\end{theorem}

\begin{proof}
We first prove $\Me \not = \emptyset \Rightarrow \mathbf{NCA(\mathcal{Y})}$.  Take $(Q^1, \dots , Q^ N) \in  \Me$ and let $(k+Y) \in  ({\sf X}_{i=1} ^{N}  K_i   + \mathcal Y ) \cap  L^{0 }_+(\Omega,  \mathbf{F}_T,P)$. Then $k^i+Y^i \geq 0$ and thus $k^i \geq -Y^i  \in L^{1 }(\Omega, \mathcal{F}^i_T , Q^i) $ for all $i$. Hence $(k^i)^- \in L^{1}(\Omega, \mathcal{F}^i_T,Q^i)$,
 $k^i \in L^{1}(\Omega, \mathcal{F}^i_T,Q^i) $ and $E_{Q^i}[k^i]=0$ (see \cite{FollmerSchied2}, Theorem 5.14). Therefore $E_{Q^i}[k^i+Y^i] = E_{Q^i}[k^i]+E_{Q^i}[Y^i]=E_{Q^i}[Y^i]$, $\sum_{i=1}^N E_{Q^i}[k^i+Y^i] = \sum_{i=1}^N E_{Q^i}[Y^i]  \leq 0$, as $(Q^1, \dots , Q^ N) \in  \Me$.
From $(k^i+Y^i) \geq 0 $ for all $i$, we also get that $\sum_{i=1}^N E_{Q^i}[k^i+Y^i]  \geq 0$, so that $\sum_{i=1}^N E_{Q^i}[k^i+Y^i] = 0$, which then implies $E_{Q^i}[k^i+Y^i] = 0$ for all $i$ and $k^i+Y^i=0$ for all $i$. Thus $\mathbf{NCA(\mathcal{Y})}$ holds true. 

\noindent We now prove $\mathbf{NCA(\mathcal{Y})} \Rightarrow \mathcal M_e(\mcY) \not = \emptyset  $.

    \textbf{Step 1} Let
    $$ \mcY=\mathrm{cone}(\{ Y_1,\dots,Y_R \}),$$ with $ Y_m \in L^{0 }(\Omega, \mathbf{F}_T,P)$, $m=1,\dots,R$, be the finitely generated convex cone containing $\mathbb R^N_0$ and let $\widehat{\mathbf P}=(\widehat{P}^1,\ldots, \widehat{P}^N)$ be the vector obtained with changes of measures given by  $$\rn{\widehat{P}^i}{P^i}=\frac{c^i}{1+\sum_{j\in (i),t}| X^j_t |+\sum_{m}|Y^i_m|}$$
     for some positive normalizing constant $c^i$. We see that $\widehat P^i \in \mathcal P_e$, $\frac {d\widehat P^i} {dP^i} \in L^{\infty }(\Omega, \mathcal{F}^i_T,P^i)$, all processes $(X^j),j \in (i)$ are integrable under $ \widehat P^i$, $ Y^i_m \in L^{1 }(\Omega, \mathcal{F}^i_T, \widehat P^i)$ for every $m$, so that $\mcY \subseteq L^{1 }(\Omega, \mathbf{F}_T, \widehat{\mathbf P})$.
    Since $\mcY$  is a finitely generated convex cone in a Hausdorff  topological vector space,  then it is closed in $L^{1 }(\Omega, \mathbf{F}_T, \widehat {\mathbf P})$, by  \cite{Aliprantis} Corollary 5.25. Moreover, from $\mcY \subseteq L^{1 }(\Omega, \mathbf{F}_T, \widehat {\mathbf P})$, we get 
    \begin{equation}\label{CY}
        K^\mathcal Y \cap L^{1 }(\Omega, \mathbf{F}_T,\widehat{\mathbf P})=  ({\sf X}_{i=1} ^{N} ( K_i - L^{0 }_+(\Omega, \mathcal{F}^i_T, \widehat P^i) )) \cap L^{1 }(\Omega, \mathbf{F}_T,\widehat{\mathbf P}) +\mcY.  
    \end{equation}
       As $\mathbf{NCA(\mathcal{Y})}$ holds,
        %holds under $P$ as well under $\widehat{P}$ and thus, 
        from Proposition \ref{C:closed} and \eqref{NCACC} we deduce that $ \widehat C^{\mcY}:=K^\mathcal Y \cap L^{1 }(\Omega, \mathbf{F}_T,\widehat {\mathbf P})$ is closed in $L^{1 }(\Omega, \mathbf{F}_T,\widehat {\mathbf P})$ and  that  $\widehat C^\mathcal Y \cap  L^{1 }_+(\Omega, \mathbf{F}_T , \widehat{{\mathbf P}})=\{0\}$. By the multidimensional version of Kreps-Yan Theorem (see \cite{BDFFM25} Theorem A.3), we deduce the existence of a vector $z=(z^1,\dots,z^N) \in L^{\infty }(\Omega, \mathbf{F}_T,\widehat {\mathbf P})$, $z^i>0$ for all $i$, such that 
    \begin{equation}\label{102}
    \sum_{i=1}^N  E_{\widehat{P}^i}[z^if^i] \leq 0 \text { for all } f \in \widehat{C}^\mathcal Y.    
    \end{equation}

    \textbf{Step 2}  
    Conditions \eqref{CY} and \eqref{102} imply $\sum_{i=1}^N  E_{\widehat{P}^i}[z^iY^i] \leq 0 \text { for all } Y \in \mathcal Y$ and from $\mathbb R^N_0 \subseteq \mcY$ we get $ E_{\widehat{P}^i}[z^i]= E_{\widehat{P}^j}[z^j]$ for all $i,j$. Set $\frac {dQ ^i} {d \widehat P^i}:=\frac {z^i} {  E_{\widehat{P}^i}[z^i]}>0$ for all $i$.  Thus $Q^i \in \mathcal P_e$ for all $i$, $\frac {dQ ^i} {d \widehat P^i} \in L^{\infty }(\Omega, \mathcal{F}^i,\widehat P^i)$, $\mcY \subseteq L^{1 }(\Omega, \mathbf{F}_T, \mathbf Q) $ for $\mathbf Q=(Q^1,\dots,Q^N)$  and $\sum_{i=1}^N  E_{Q^i}[Y^i] \leq 0 \text { for all } Y \in \mathcal Y$.
    \noindent Moreover, \eqref{CY} and \eqref{102} imply $\sum_{i=1}^N  
    E_{Q^i}[k^i] \leq 0 \text { for all } k \in ({\sf X}_{i=1} ^{N} ( K_i - L^{0 }_+(\Omega, \mathcal{F}^i_T, \widehat P^i) )) \cap L^{1 }(\Omega, \mathbf{F}_T,\widehat {\mathbf P} )$, so that $ E_{Q^i}[k^i] = 0 \text { for all } k^i \in K_i \cap L^{1 }(\Omega, \mathcal{F}^i_T, \widehat P^i), $ for all $i$. 
    As $ X^{j} \subseteq L^{1 }(\Omega, \mathcal{F}^i_T, \widehat P^i)$, for all $j \in (i)$ and all $i$, we recall from  \cite{DS2006}, Section 6.11, or \cite{FollmerSchied2} Theorem 5.14 that $\mib(\widehat P)$ can be written as  
\begin{equation}\label{MartingaleMeasures}
\mib(\widehat P^i)=\left\{ Q \in \mathcal{P}_{ac} \mid \frac {dQ} {d\widehat P^i} \in L^{\infty }(\Omega, \mathcal{F}^i_T ,\widehat P^i)  \text { and } E_Q[k]= 0  \text {  } \forall k \in \ki \cap L^{1 }(\Omega, \mathcal{F}_T^{i},\widehat P^i) \right\}.
\end{equation}
Thus we conclude that $Q^i \in \mibe(\widehat P^i)$, for all $i$. Also, since $\frac{dQ^i}{d{P}^i}=\frac{dQ^i}{d\widehat{P}^i} \frac{d\widehat P^i}{d{P}^i}$ and both terms in right-hand side belong to $ L^{\infty}(\Omega,\mathcal{F}^i_T, P)=L^{\infty}(\Omega,\mathcal{F}^i_T, \widehat P^i)=L^{\infty}(\Omega,\mathcal{F}^i_T, P^i)$, we deduce $\frac{dQ^i}{d{P}^i}\in L^{\infty}(\Omega,\mathcal{F}^i_T, P^i)$.  Moreover,  if $k \in \ki \cap L^{1 }(\Omega, \mathcal{F}_T^{i}, P^i)$, then  $k \in \ki \cap L^{1 }(\Omega, \mathcal{F}_T^{i},\widehat P^i)$, so that $Q^i \in \mibe(P^i)$. 
We conclude that $\mathcal M_{e}^{\infty}(\mcY,\mathbf{P})\neq \emptyset$. 

In particular also 
 $ \emptyset \neq \mathcal M_{e}^{\infty}(\mcY,\widehat{\mathbf{P}})\subseteq  \Me.$ which was the only thing left to prove.
\end{proof}

\subsection{Collective superreplication}

\begin{definition}
For a random variable $f \in L^{0}(\Omega, \mathcal{F}^i_T , P)$, we define the classical superreplication  price for agent $i$ as
\begin{equation*}
\pii(f):=\inf\{x \in \mathbb{R} \mid \exists k \in \ki \text{ s.t. } x+k \geq f \}.
\end{equation*}
\end{definition}
\noindent Under $\mathbf{NA}_i$, the following classical superhedging duality holds true 
%(see \cite{BDFFM25} page 12). 
\begin{equation}\label{superclassic}
\pii(f)=\sup_{Q \in M^i} E_{Q}[f] \ \text { for all } f \in L^{\infty}(\Omega, \mathcal{F}^i_T , P).
\end{equation}

In order to state the collective version of the pricing-hedging duality, we recall from \cite{BDFFM25} the concept of collective super/subreplication price.

\begin{definition}[Definition 4.1 and 4.16 \cite{BDFFM25}]\label{defsupsub}
    For $\mcY\subseteq L^{0}(\Omega, \mathbf{F}_T,P)$ and $(f^1,\dots,f^N)=f\in L^{0}(\Omega, \mathbf{F}_T, P)$ we define the collective superhedging and subhedging price as
\begin{align}
    \label{super:rho} \rho^{\mcY}_+(f) & =\inf\left\{\sum_{i=1}^N x^i \mid x\in\R^N \text{ and }  x+k+Y \geq f \;\text{for }  k\in {\sf X}_{i=1} ^{N}  K_i,\, Y\in \mcY \right\},
    \\ \notag \rho^{\mcY}_-(f) & =\sup\left\{ \sum_{i=1}^N x^i\mid x\in\R^N \text{ and }  x+k-Y \leq f \;\text{for }  k\in {\sf X}_{i=1} ^{N}  K_i,\, Y\in \mcY \right\}.
\end{align}

The classical superreplication price of the $N$ claims $f$ is defined as $\rho^{{N}}_+(f):=\rho^{{\mcY}}_+(f)$ for $\mcY=\{0\}$. 
\end{definition}

We defer the reader to \cite{BDFFM25} Section 4 for the interpretation of $\rho^{\mcY}_+(f)$ and the difference with the classical superhedging price $\rho^{{N}}_+(f)$. We also note that for the convex cone $\mcY$ it holds
\begin{equation}
\label{rem:rhopm}
    \rho^{\mcY}_-(f)=-\rho^{\mcY}_+(-f)\text{ for every }f\in L^{0}(\Omega, \mathbf{F}_T,P).
\end{equation}

\begin{remark}\label{remMinus}
The minus sign $-Y$ appearing in the definition of the subhedging price is essential when the set of admissible exchanges $\mcY$ is a \emph{convex cone} rather than a vector space. Indeed, this sign is precisely what ensures that \eqref{rem:rhopm} remains valid, in analogy with the classical single-agent framework.

To clarify its financial interpretation (in the individual agent case), consider first the simplest situation in which no hedging opportunities are available. The selling price $p_+(f)$ is the smallest amount $x\in\mathbb R$ that an agent is willing to accept in order to sell the claim $f$, namely
\[
p_+(f):=\inf\{x\in\mathbb R\mid x-f\ge 0\}
=\operatorname*{ess\,sup}(f).
\]
Similarly, the buyer price $p_-(f)$ is the largest amount $x\in\mathbb R$ that the agent is willing to pay in order to purchase $f$, that is
\[
p_-(f):=\sup\{x\in\mathbb R\mid f-x\ge 0\}
=\operatorname*{ess\,inf}(f).
\]
Naturally, the no-loss requirement encoded in the almost sure inequalities above may be replaced by alternative acceptability criteria, leading for instance to indifference selling and buyer prices. The economic interpretation, however, remains the same.

Suppose now that trading is allowed in both a vector space  $\mathbb K$, representing a liquid  \textquotedblleft market", and through other admissible investments belonging to a cone $\mathbb Y$. Then the selling price is the smallest amount $x\in\mathbb R$ that the agent is willing to receive in order to sell $f$, after taking suitable positions in $\mathbb K$ and $\mathbb Y$:
\[
p_+(f)
:=
\inf\left\{
x\in\mathbb R \;\middle|\;
x-f+k+Y\ge 0
\text{ for some } k\in\mathbb K,\; Y\in\mathbb Y
\right\}.
\]
Likewise, the buyer price $p_-(f)$ is the largest amount $x\in\mathbb R$ that the agent is willing to pay in order to buy $f$, while still being able to invest in $\mathbb K$ and $\mathbb Y$:
\begin{align*}
p_-(f)
&:=
\sup\left\{
x\in\mathbb R \;\middle|\;
f-x+k+Y\ge 0
\text{ for some } k\in\mathbb K,\; Y\in\mathbb Y
\right\} \\
&=
\sup\left\{
x\in\mathbb R \;\middle|\;
f \geq x+k-Y
\text{ for some } k\in\mathbb K,\; Y\in\mathbb Y
\right\}.
\end{align*}
In the last inequality, we replaced, w.l.o.g., $-k$ with $k$ because $\mathbb K$ is a vector space. By contrast, the minus sign in front of $Y$ cannot be removed, since $\mathbb Y$ is only assumed to be a cone.
This shows that the superhedging and subhedging are genuinely distinct notions when $\mcY$ is a cone and explains the two expressions in Definition \ref{defsupsub}. 
\end{remark}

To handle the superhedging duality of potentially unbounded or non-integrable contingent claim vectors $f \in L^0(\Omega, \mathbf{F}, P)$ we introduce the following convex set of vectors of martingale measures. For any arbitrary $(\varphi^1,\dots,\varphi^N)=\varphi\in L^{0}(\Omega, \mathbf{F}_T,P)$ let
 \begin{equation}\label{Mfi}\mathcal{M}^{\varphi}_e(\mcY):=\{\mathbf Q\in \mathcal{M}_e(\mcY)\mid E_{Q_i}[|\varphi^i|]<\infty\;\forall\,i=1,\ldots,N\} \subseteq \mathcal{M}_e(\mcY) .
 \end{equation}
 Recalling Remark \ref{remPP}, we observe that for $\mcY=\mathrm{cone}(\{Y_1,\dots, Y_R\})$ and  $\varphi^i=\max \{|Y^i_1|,\ldots, |Y^i_R|\}$ we have $\mathcal{M}^{\varphi}_e(\mcY)=\Me $. \\

We shall use the set \(\mathcal M_e^{\varphi}(\mathcal Y)\) in the following way. Let
$f\in L^0(\Omega,\mathbf F_T,P)$
and suppose that $f=x+k+Y$, for some
$x\in\mathbb R^N,\,k\in {\sf X}_{i=1}^{N}K_i,
\,Y\in\mathcal Y$.
If \(\varphi\in L^0(\Omega,\mathbf F_T,P)\) is such that
$\max\big\{|f^i|,\ |Y^i_1|,\ldots,|Y^i_R|\big\}\le \varphi^i\, \text{for each } i,$
then for every $\mathbf Q\in\mathcal M_e^{\varphi}(\mathcal Y)$ it follows that
$f\in L^1(\Omega,\mathbf F_T,\mathbf Q)$
and $\mathcal Y\subseteq L^1(\Omega,\mathbf F_T,\mathbf Q).$
Hence
$k \in L^1(\Omega,\mathbf F_T,\mathbf Q),$
which implies $E_{Q^i}[k^i]=0
\, \text{for all } i=1,\ldots,N.$

The following result is akin to Theorem 4.12 in \cite{BDFFM25}, where the pricing-hedging duality was proved for a general convex cone $\mcY$ under $\mathbf{NCA(\mathcal{Y})}$, but with the additional requirement of $K^{\mcY}$ being closed in probability. Thanks to Proposition \ref{C:closed}, such a property follows directly from $\mathbf{NCA(\mathcal{Y})}$ leading to a genuine extension of \cite{DFM25} Theorem 2.7  to the case of a convex cone $\mcY$. 

\begin{theorem}[Pricing-Hedging duality]\label{duality} Let Assumption \ref{ass:cone} hold true for $\mcY $ be given in \eqref{coY} and take $f\in L^{0}(\Omega, \mathbf{F}_T,P)$.  Under  $\mathbf{NCA(\mathcal{Y})}$,  $\rho^{\mcY}_+(0)=0$ and, for any $\varphi\in L^{0}(\Omega, \mathbf{F}_T,P)$ such that  
$\max \{|f^i|,|Y^i_1|,\ldots, |Y^i_R|\}\leq \varphi^i$ for every $i=1,\dots,N$,  we have $$\mathcal{M}^{\varphi}_e(\mcY) \neq \emptyset$$ and
\begin{equation}\label{pricing:hedging}
  \rho^{\mcY}_+(f)=\sup \left \{\sum_{i=1}^N E_{Q^i}[f^i]\mid \mathbf{Q}\in \mathcal{M}^{\varphi}_e(\mcY) \right \}>-\infty. 
\end{equation}
  Moreover, when finite, $ \rho^{\mcY}_+(f) $ defined in \eqref{super:rho} is attained by some $x \in \R^N$.  
\end{theorem}

\begin{proof}
The condition $\rho^{\mcY}_+(0)=0$ was established in \cite{BDFFM25} Lemma 4.4. 
Take $\widehat{P}$ as in Remark \ref{remPP}, so that $\mcY \subseteq L^{1}(\Omega, \mathbf{F}_T,\widehat{P})$.  The condition $\mathbf{NCA(\mathcal{Y})}$ holds equivalently under $P$ or $\widehat{P}$ and by Proposition \ref{C:closed} the cone $K^{\mcY}$ is closed in  probability. All the assumptions of Theorem 4.12 in \cite{BDFFM25} are satisfied and thus, from \cite{BDFFM25} eq. (4.9),
\begin{equation}\label{rho23}
  \rho^{\mcY}_+(f)=\sup \left \{\sum_{i=1}^N E_{Q^i}[f^i]\mid \mathbf{Q}\in \mathcal{M}^{\infty}(\mcY,\widehat{P}) \right \}>-\infty, 
\end{equation}
where $\mathcal{M}^{\infty}(\mcY,\widehat{P})$ is  
the collection of $(Q^1,\ldots,Q^N)\in {\sf X}_{i=1} ^{N}M^{i,\infty}(\widehat{P})$ such that 
$\sum_{i=1}^N E_{Q^i}[Y^i]\leq 0$   for all $Y \in \mcY $. By $\mathbf{NCA(\mathcal{Y})}$ and \eqref{MMM}, 
let $\mathbf Q_e\in\mathcal{M}^{\infty}_e(\mcY,\widehat{P}) \neq \emptyset$. By Remark \ref{remPP}, $\varphi\in L^1(\Omega, \mathbf{F}_T,\mathbf Q)$ for all $\mathbf Q=(Q^1,\ldots,Q^N)\in {\sf X}_{i=1} ^{N}M^{i,\infty}(\widehat{P})$. Recalling that $\mathcal{M}^{\varphi}_e(\mcY)$ is the collection of $\mathbf Q=(Q^1,\ldots,Q^N)\in {\sf X}_{i=1} ^{N}M^i_e$ such that $\varphi\in L^1(\Omega, \mathbf{F}_T,\mathbf Q)$ and   
$\sum_{i=1}^N E_{Q^i}[Y^i]\leq 0$ for all $Y \in \mcY$, we obtain that $\mathbf Q_e \in \mathcal{M}^{\varphi}_e(\mcY)$.
Moreover for any $\mathbf Q\in \mathcal{M}^{\infty}(\mcY,\widehat{P})$ and $\lambda\in (0,1)$ we have $\lambda \mathbf Q_e+(1-\lambda) \mathbf Q\in \mathcal{M}^{\infty}_e(\mcY,\widehat{P})$. Set $\frac{dQ_{\lambda}^i}{d\widehat{P}}:=(1-\lambda) \frac{dQ^i}{d\widehat{P}}+\lambda\frac{dQ^i_e}{d\widehat{P}}$ for any $i=1,\ldots,N$, then   $\frac{dQ_{\lambda}^i}{d\widehat{P}}\to \frac{dQ^i}{d\widehat{P}}$ as $\lambda\to 0$ with respect to the norm $\|\cdot\|_{\infty}$ on $L^{\infty}(\Omega,\mcF_T^i,\widehat{P})$. This automatically implies
\[\left(\frac{dQ_{\lambda}^1}{d\widehat{P}},\ldots, \frac{dQ_{\lambda}^N}{d\widehat{P}}\right)\to \left(\frac{dQ^1}{d\widehat{P}},\ldots, \frac{dQ^N}{d\widehat{P}}\right)\]
with respect to any equivalent norm on the product space $L^{\infty}(\Omega, \mathbf{F}_T,\widehat{P})$, for example
$\|Z\|_{\infty,N}= \sup_{i=1,\ldots,N}\|Z^i\|_{\infty}$.
As the map 
\[\left(\frac{dQ^1}{d\widehat{P}},\ldots, \frac{dQ^N}{d\widehat{P}}\right)\mapsto \sum_{i=1}^N E_{\widehat{P}}\left[\frac{dQ^i}{d\widehat{P}}f^i\right]\]
is $\|\cdot\|_{\infty,N}$ continuous for any $f\in L^1(\Omega,\mathbf{F}_T,\widehat{P})$, then from \eqref{rho23}
\begin{align*}
\rho^{\mcY}_+(f)&=\sup \left \{\sum_{i=1}^N E_{\widehat{P}}\left[\frac{dQ^i}{d\widehat{P}}f^i\right]\mid \mathbf Q\in \mathcal{M}^{\infty}(\mcY,\widehat{P}) \right \}\\
 & =  \sup \left \{\sum_{i=1}^N E_{Q^i}[f^i]\mid \mathbf Q\in \mathcal{M}^{\infty}_e(\mcY,\widehat{P}) \right \}
  \\& =  \sup \left \{\sum_{i=1}^N E_{Q^i}[f^i]\mid \mathbf Q\in \mathcal{M}^{\infty}_e(\mcY,\widehat{P}) \text{ and } \varphi\in L^1(\Omega, \mathbf{F}_T,\mathbf Q)\right \}\\
    % \end{align*}
    %   \begin{align*}
   & \leq  \sup \left \{\sum_{i=1}^N E_{Q^i}[f^i]\mid \mathbf Q\in \mathcal{M}^{\varphi}_e(\mcY) \right \}
  \\ & \leq  \inf\left\{\sum_{i=1}^N x^i \mid x\in\R^N \text{ and }  x+k+Y \geq f \;\text{for }  k\in {\sf X}_{i=1} ^{N}  K_i,\, Y\in \mcY \right\}
  \\ & = \rho^{\mcY}_+(f)
  \end{align*}

Now we proceed proving the attainment when $\rho^{\mcY}_+(f)$ is finite. By definition of $\rho^{\mcY}_+(f)$, there exists a sequence $(x_n)_n$ in  $\mathbb R^N$ such that $\widehat x_n=:\sum_{i=1}^N x^i_n \downarrow \mcR (f)$ and $f-x_n \in K^\mathcal Y$ (as defined in \eqref{KK}). Letting $\mathbf{1}=(1,\dots,1)$ we see that $(\frac{\widehat x_n}{N}\mathbf 1-x_n) \in \R^N_0 \subseteq K^\mathcal Y$. From  $f-x_n=f-\frac{\widehat x_n}{N}\mathbf 1 +\frac{\widehat x_n}{N} \mathbf 1-x_n $, we see that also $f-\frac{\widehat x_n}{N}\mathbf 1 \in  K^\mathcal Y$. 
From the closure of $K^\mathcal Y$ and $f-\frac{\widehat x_n}{N}\mathbf 1 \rightarrow f-\frac{\mcR (f)}{N}\mathbf 1$ we obtain $f-\frac{\mcR (f)}{N}\mathbf 1 \in K^\mathcal Y$. 
Thus $x:=\frac{\mcR (f)}{N}\mathbf 1 \in \R^N$ attains the infimum.

\end{proof}

\begin{remark}
    Applying the formula \eqref{pricing:hedging} when  $\mcY=\R^N_0$, we obtain for $f\in L^{0}(\Omega, \mathbf{F}_T,P)$
    \begin{equation}\label{rhoNN}
        \rho^{{\mcY}}_+(f)=\rho^{{N}}_+(f)=\sum_{i=1}^N \sup \left \{ E_{Q^i}[f^i]\mid Q^i\in M^i_e \,\text{ and }\, E_{Q_i}[|\varphi^i|]<\infty \right \},
    \end{equation}as in this case $\Me=M^1_e \times \dots \times M^N_e$. 
\end{remark}

\begin{remark}\label{monotonia}
Let Assumption \ref{ass:cone} and $\mathbf{NCA(\mathcal{Y})}$ hold true and let $f\in L^{0}(\Omega, \mathbf{F}_T,P)$. By simple properties of convex functional vanishing at $0$, it is easy to show, see for instance \cite{FrittelliRosazza} Remark 8 l), that  $ -\rho^{{\mcY}}_+(-f) \leq \rho^{{\mcY}}_+(f)$. Therefore,  $\rho^{ \mcY}_-(f) \leq\rho^{ \mcY}_+(f)$, for every $f\in L^{0}(\Omega, \mathbf{F}_T,P)$.
\end{remark}

%%%%%%%%%%%%%%%%%%%%%%%%%%%%%%%%%%%
\section{Collectively replicable claims and market completeness}\label{completeness}

In this section we investigate the notion of replicability in the collective market introduced above, with particular emphasis on the effects generated by the fact that the set of admissible exchanges $\mathcal Y$ is now a convex cone rather than a vector space. This distinction is not merely technical: when $\mathcal Y$ is a vector space as in \cite{DFM25}, every admissible exchange can be reversed, and the usual symmetry between superhedging and subhedging is preserved. By contrast, when $\mathcal Y$ is only a cone, admissible exchanges are directional: an exchange $Y\in\mathcal Y$ need not be accompanied by its opposite $-Y$. As a consequence, collective replicability alone no longer guarantees equality between superhedging and subhedging prices. In Proposition \ref{diversi} we 
clarify the precise relation between collective replicability and the coincidence of collective superhedging and subhedging prices. In particular, we show that price equality forces $f$ to admit both a positive and a negative exchange representation, while the converse requires the additional viability condition $\mathbf{NCA}(-\mathcal Y)$.

 Secondly we recall that the $\mathbf{NCA}(\mathcal{Y})$-prices of vectors $f=(f^1,\dots,f^N) \in L^{0 }(\Omega, \mathbf{F}_T,P)$ of contingent claims can be characterized  as expectations under equivalent collective martingale measures (see Theorem \ref{BaseFTAPII}). In Proposition \ref{propRO}, we establish that, under $\mathbf{NCA}(\mathcal{Y})$ and when $\mcY$ is a vector space, the set $\Pi(f) \subseteq \R^N$ of such $\mathbf{NCA}(\mathcal{Y})$-prices is always a relatively open nonempty convex set, regardless of whether $f$ is replicable or not. This is in agreement with the classical theory for $N=1$, as in this case relatively open nonempty convex sets are either open (and this occurs for non replicable claims) or are a singleton (and this occurs for replicable claims).
As discussed in Remark \ref{remrelopen},
 this property may fail when $\mcY$ is merely a convex cone.

\medskip

To recover a genuine equivalence between replicability and price uniqueness, we introduce the restricted cone $\widehat{\mathcal Y}$, consisting of exchanges whose individual components are fairly valued under every collective equivalent martingale measure, rather than merely requiring the aggregate exchange to be fairly assessed.
This restriction eliminates the ambiguity caused by exchange components with nonzero individual values. We therefore call a claim strongly collectively replicable if it is replicable using exchanges in $\widehat{\mathcal Y}$. The main result of this section shows that strong collective replicability is exactly the appropriate multi-agent analogue of classical replicability. Finally, we use these results to formulate a collective version of the Second Fundamental Theorem of Asset Pricing for convex cones of exchanges. The theorem characterizes collective completeness, strong collective completeness, uniqueness of the collective equivalent martingale measure, and equality between superhedging and subhedging functionals,  under the appropriate no-arbitrage assumptions for both $\mathcal Y$ and $-\mathcal Y$.

\subsection{Collectively replicable claims}

\begin{definition}[\cite{DFM25} Definition 2.9]\label{defrep}
 Consider the vector of contingent claims $ f=( f^1,\dots,  f^N) \in L^{0}(\Omega, \mathbf{F}_T, P)$, $N \geq 1$ and let $\mcY \subseteq L^{0 }(\Omega, \mathbf{F}_T , P)$ be a convex cone. We say that:
 \begin{enumerate}
     \item $f$ is $\mcY$-collectively replicable if $ f\in \R^N+ {\sf X}_{i=1} ^{N} K_i+\mathcal Y$   \item $ f$  is classically replicable if it is $\mcY$-collectively replicable for $\mcY=\{0\}$, namely if $ f\in \R^N+ {\sf X}_{i=1} ^{N} K_i  $.
 \end{enumerate}
  
\end{definition}

\begin{proposition}[\cite{DFM25} Proposition 3.4]
\label{uguali}
 Assume $\mcY$ is a finite dimensional vector space containing $\mathbb R^N_0$ and for which $\mathbf{NCA}(\mcY)$ holds. Then $f \in L^{0 }(\Omega, \mathbf{F}_T , P)$ is $\mcY$-collectively replicable if and only if $\rho^{ \mcY}_-(f)=\rho^{ \mcY}_+(f)$.
\end{proposition}

For vector spaces \(\mathcal Y\), there is no need to distinguish between \(\mathbf{NCA}(\mathcal Y)\) and \(\mathbf{NCA}(-\mathcal Y)\), as the two conditions are equivalent. Similarly, a claim \(f\) is collectively \(\mathcal Y\)-replicable if and only if \((-f)\) is. As we shall see, this symmetry no longer holds when \(\mathcal Y\) is only assumed to be a convex cone, and these distinctions become relevant.

\begin{remark}
Notice that both $ f$ and $(-f)$ are $\mcY$-collectively replicable if and only if $$f\in (\R^N + {\sf X}_{i=1} ^{N} K_i +\mcY) \cap (\R^N + {\sf X}_{i=1} ^{N} K_i -\mcY).$$  
Indeed, if the condition holds then clearly $f$ is $\mcY$ collectively replicable, but also $f=x+k-Y$ with obvious notation, so that $-f=(-x)+(-k)+Y\in \R^N + {\sf X}_{i=1} ^{N} K_i +\mcY$, i.e. $-f$ is $\mcY$-collectively replicable. Conversely, if both $\pm f$ are $\mcY$-collectively replicable then  $f\in \R^N + {\sf X}_{i=1} ^{N} K_i +\mcY$ by definition, and $-f=x+k+Y\in \R^N + {\sf X}_{i=1} ^{N} K_i +\mcY$. Hence $f=(-x)+(-k)-Y\in \R^N + {\sf X}_{i=1} ^{N} K_i -\mcY$, and $f$ belongs to the desired intersection
\end{remark}

Suppose that $\mathcal{Y}$ is a convex cone, rather than a vector space, satisfying Assumption \ref{ass:cone}.
Then the conclusion in Proposition \ref{uguali} fails. Indeed, by Item 1 in the following proposition, if $f$ is $\mcY$-collectively replicable  but $(-f)$ is not, then a bid-ask spread necessarily arises, that is,
$\rho^{ \mcY}_-(f)<\rho^{ \mcY}_+(f)$.
Moreover, the condition $\mathbf{NCA}(\mcY)$ alone is not any more sufficient to guarantee that $\rho_-^{\mathcal Y}(f)=\rho_+^{\mathcal Y}(f)$, even if both $f$ and $(-f)$ are $\mcY$-replicable claims. A counterexample is described in Section \ref{example1},  Example \ref{example42}.

\begin{proposition}\label{diversi} Suppose  Assumption \ref{ass:cone} holds true and $f \in L^{0 }(\Omega, \mathbf{F}_T , P)$.
\begin{enumerate}
    \item Assume $\mathbf{NCA}(\mcY)$ holds. If
$\rho^{ \mcY}_-(f)=\rho^{ \mcY}_+(f)$, then $$f\in (\R^N + {\sf X}_{i=1} ^{N} K_i +\mcY) \cap (\R^N + {\sf X}_{i=1} ^{N} K_i -\mcY),$$ i.e. both $ f$ and $(-f)$ are $\mcY$-collectively replicable.
\item Assume $\mathbf{NCA}(\mcY)$ and $\mathbf{NCA}(-\mcY)$ hold. 
Suppose that both $ f$ and $(-f)$ are $\mcY$-collectively replicable. 
Then $$\rho^{ \mcY}_-(f) = \rho^{ \mcY}_+(f).$$ 
\end{enumerate}
\end{proposition}

\begin{proof} 
To prove item 1., assume $\rho^{ \mcY}_-(f)=\rho^{ \mcY}_+(f)$. Then they are both finite. Indeed from Theorem \ref{duality} $\rho^{ \mcY}_+(f)>-\infty$ and $\rho^{ \mcY}_{-}(f)<+\infty$. From the attainment in Theorem \ref{duality},  there exists $\tilde x\in \R^N$ with $\sum_i\tilde x^i=\rho^{ \mcY}_-(f)$, $\tilde k\in {\sf X}_{i=1} ^{N} K_i$ and $\tilde Y\in \mcY$ such that $f \geq \tilde x+\tilde k-\tilde Y $ and there exists  $k\in {\sf X}_{i=1} ^{N} K_i$, $Y\in \mcY$, and $x\in \R^N$ with $\sum_i x^i=\rho^{ \mcY}_+(f)$ such that  $ x+k+Y \geq f$. Therefore $x+k+Y \geq f \geq \tilde x+\tilde k-\tilde Y$, that is:
 \[x-\tilde x + k-\tilde k +Y+\tilde Y\geq 0 \text{ with } \sum_{i=1}^N (x^i-\tilde x^i)=0\]
 which implies, by $\mathbf{NCA}(\mcY)$ and since $x-\tilde x  +Y+\tilde Y\in \mcY$, that
 $x-\tilde x + k-\tilde k +Y+\tilde Y = 0$. Thus 
 $$\tilde x+\tilde k-\tilde Y=x+k+Y\geq f \geq \tilde x+\tilde k-\tilde Y=x+k+Y, $$ 
 so that $f\in (\R^N + {\sf X}_{i=1} ^{N} K_i -\mcY)\cap(\R^N + {\sf X}_{i=1} ^{N} K_i +\mcY)$, which is our thesis.

To prove item 2., we
suppose that $f\in (\R^N + {\sf X}_{i=1} ^{N} K_i -\mcY)\cap(\R^N + {\sf X}_{i=1} ^{N} K_i +\mcY)$. Then with obvious notation $\tilde x+\tilde k-\tilde Y= f =x+k+Y$. We show that $\sum_{i=1}^N \tilde{x}^i\geq \sum_{i=1}^N x^i$. Suppose indeed by contradiction that $\sum_{i=1}^N\tilde x^i<\sum_{i=1}^N x^i$ and set $z^i=\frac1N (\sum_{i=1}^N \tilde x^i-\sum_{i=1}^N x^i)<0$. Then
$$ 0= x-\tilde x + k-\tilde k +Y+\tilde Y> (x-\tilde x + z)+ (k-\tilde k) +Y+\tilde Y \in {\sf X}_{i=1} ^{N} K_i +\mcY$$
where the strict inequality is meant to hold componentwise a.s., and the last claim is a consequence of $x-\tilde x+z\in\R^N_0\subseteq \mcY$. 
Changing signs we have $$0< (-k+\tilde k) -\left((x-\tilde x +z)+Y+\tilde Y\right)\in {\sf X}_{i=1} ^{N} K_i +(-\mcY).$$
This contradicts $\mathbf{NCA}(-\mcY)$. 

 Thus, combining with $\rho^{ \mcY}_+(f) \geq \rho^{ \mcY}_-(f)$ (which holds under $\mathbf{NCA}(\mcY)$, by Remark \ref{monotonia}) we get
$$\sum_{i=1}^N\tilde x^i \geq \sum_{i=1}^N x^i\geq \rho^{ \mcY}_+(f) \geq \rho^{ \mcY}_-(f)\geq \sum_{i=1}^N\tilde x^i.$$
This yields
\begin{equation}\label{sumx}
    \sum_{i=1}^N\tilde x^i =\rho^{ \mcY}_-(f) = \rho^{ \mcY}_+(f)= \sum_{i=1}^N x^i.
\end{equation}
\end{proof}
We note that in the previous result, the proof of Item 2 actually provided further insights on the links between prices of collectively replicating strategies and sub\slash subhedging prices. We highlight these in the following Corollary.
\begin{corollary}
\label{corexp0}
    Suppose  Assumption \ref{ass:cone},  $\mathbf{NCA}(\mcY)$ and $\mathbf{NCA}(-\mcY)$ hold. If both $ f$ and $(-f)$ are $\mcY$-collectively replicable, i.e. there exist $x, \widetilde{x}\in\R^N, k,\widetilde{k}\in{\sf X}_{i=1}^N K_i,Y,\widetilde{Y}\in\mcY$ such that $\widetilde{x}+\widetilde{k}-\widetilde{Y}=f=x+k+Y$, then
\begin{enumerate}
 \item 
 $\sum_{i=1}^N\tilde x^i =\rho^{ \mcY}_-(f) = \rho^{ \mcY}_+(f)= \sum_{i=1}^N x^i$;
 
\item 
$\sum_{i=1}^NE_{Q^i}[\widetilde{Y}^i]=0=\sum_{i=1}^NE_{Q^i}[{Y}^i]\, \text{ for all } \mathbf Q\in \mathcal{M}_e(\mcY)$;

\item For any $\varphi\in L^{0}(\Omega, \mathbf{F}_T,P)$ such that  
$\max \{|f^i|,|Y^i_1|,\ldots, |Y^i_R|\}\leq \varphi^i$ for every $i=1,\dots,N$ it holds that
$$\rho^{\mcY}_-(f)=\sum_{i=1}^N E_{Q^i}[f^i]=\rho^{\mcY}_+(f)\quad \forall \mathbf{Q}\in \mathcal{M}^{\varphi}_e(\mcY).$$ 

\end{enumerate}
\end{corollary}

\begin{proof}
    The Item 1 follows from \eqref{sumx}. 
    Item 3 is a consequence of Items 1 and 2. It only remains to prove Item 2.
    Take $\mathbf Q\in \mathcal{M}_e(\mcY)$. Observe that since $\mcY \in L^1(\Omega,\mathbf{F}_T,\mathbf{Q})$, we also deduce $k-\widetilde{k}=\widetilde{x}-x-(\widetilde{Y}+Y)\in L^1(\Omega,\mathbf{F}_T,\mathbf{Q})$, 
    and from $Q^i\in M^i_e$ we get $E_{Q^i}[k^i-\widetilde{k}^i]=0$ for every $i$. 
Now we can write
$$-\widetilde{Y}=x-\widetilde{x}+k-\widetilde{k}+Y,$$
take expectations under \(\mathbf Q\) and sum componentwise. As $\sum_{i=1}^N\widetilde{x}^i=\sum_{i=1}^N x^i$, $E_{Q_i}[k^i-\widetilde{k}^i]=0$ for all $i$,  \(Y\in \mcY\) and \(\mathbf Q\in \mathcal M_e(\mcY)\) we get
\[
0\leq -\sum_{i=1}^N E_{Q_i}[\widetilde{Y}^i]
=
\sum_{i=1}^N x^i-\sum_{i=1}^N\widetilde{x}^i
+
\sum_{i=1}^N E_{Q_i}[k^i-\widetilde{k}^i]
+
\sum_{i=1}^N E_{Q_i}[Y^i]= 
\sum_{i=1}^N E_{Q_i}[Y^i]\leq 0.
\]
We conclude that
$
\sum_{i=1}^N E_{Q_i}[\widetilde{Y}^i]=0=\sum_{i=1}^N E_{Q_i}[Y^i]
$.

\end{proof}

\subsection{On the set of \textbf{NCA}$(\mcY)$-prices of contingent claims}

Recall that the (global) securities market comprises a zero-interest rate riskless asset $X^0$ and $J$ risky assets with discounted price processes $X^j=(X^j_t)_{t\in \mathcal T}$, $j=1, \dots ,J$, $J\geq 1$. 

\begin{definition}[\cite{DFM25} Definition 3.1]
A vector $\Pi_ f \in \mathbb{R}^N$ is called a  \textbf{NCA}$(\mcY)$-price for the contingent vector $f=(f^1,\dots,f^N) \in L^{0 }(\Omega, \mathbf{F}_T,P)$ if there exist processes $X^{J+1},\dots, X^{J+N}$ such that:
\begin{enumerate}
\item $X^{J+i}$ is adapted  to the  filtration $\fib$, $i=1,\dots, N$;
    \item $\Pi^i_ f = X^{J+i}_0$, $i=1,\dots, N$;
    \item $X^{J+i}_T =    f^i$, $i=1,\dots, N$;
    \item  \textbf{NCA}$(\mcY)$ holds in the extended market: $(X^j)_{j\in (i)\cup\{J+i\}}, \, i=1,\dots,N.$
\end{enumerate} 
Let $\Pi( f)$ be the set of \textbf{NCA}$(\mcY)$-prices for the contingent claims $ f=( f^1,\dots,  f^N).$ 
 \end{definition}
The following result holds for a convex cone $\mcY$ and is an improved version of Theorem 3.2 in \cite{DFM25}. The expression in \eqref{45} will be useful in proving that $\Pi(f)$ is relatively open.
Recall that if $\mathbf{\widehat{P}}=(\widehat{P}^1,\dots,\widehat{P}^N) \in (\mathcal P_{e})^{N} $ we set $ L^{p }(\Omega, \mathbf{F}_T,\mathbf{\widehat{P}}):=L^{p }(\Omega, \mathcal{F}_T^{1},\widehat{P}^1) \times \dots \times  L^{p }(\Omega, \mathcal{F}_T^{N},\widehat{P}^N)$.
\\

% %%%%%%%%%%%%%%%%%%%%%%%%%

%%%%%%%%%%%%%%%%%%%%%%%%%%%%

%%%%%%%%%%%%%%%%%%%%%%%%%%
%\red{INIZIO NUOVA FORMULAZIONE}

\begin{theorem}
\label{BaseFTAPII}
 Suppose that Assumption \ref{ass:cone} and $\mathbf{NCA}(\mcY)$ hold true and let $({P}^1,\dots,{P}^N)$ be any vector of probability measures such that ${P}^i \sim P$ for all $i$. 
Then, for any $ f\in L^{0}(\Omega, \mathbf{F}_T,P)$, and any fixed $\varphi\in L^{0}(\Omega, \mathbf{F}_T,P)$ such that  
$\max \{|f^i|,|Y^i_1|,\ldots, |Y^i_R|\}\leq \varphi^i$, it holds that
\begin{align}
\Pi( f) & = \left\{ (E_{Q^1}[   f^1],\dots, E_{Q^N}[   f^N]) \mid \mathbf{Q} \in \mathcal{M}^\varphi_e(\mcY)\right\} \label{44} \\ 
&=\left\{ (E_{Q^1}[   f^1],\dots, E_{Q^N}[   f^N]) \mid \mathbf{Q} \in \mathcal{M}^\varphi_e(\mcY) \text{ s.t. } \frac{dQ^i}{d{P}^i} \in L^{\infty}(\mathcal{F}^i_T) \text{ for all } i \right\},  \label{45}  
\end{align}
where $\mathcal{M}^\varphi_e(\mcY)$ is given in \eqref{Mfi} and 
%$ \frac{d\widehat{P}^i}{dP} \in L^{\infty}(\mathcal{F}^i_T)$, for all $i$, for
$L^{\infty}( \mathcal{F}^i_T):=L^{\infty}(\Omega, \mathcal{F}^i_T,{P}^i)=L^{\infty}(\Omega, \mathcal{F}^i_T,P)$.
\end{theorem}

\begin{proof}
Suppose first that $\Pi_ f \in \Pi( f)$. Then \textbf{NCA}$(\mcY)$ holds in the extended market. It is always possible to select probability measures $\mathbf{\widehat{P}}=(\widehat{P}^1,\dots,\widehat{P}^N)$,  all equivalent to $P$, with bounded densities $\frac{d\widehat{P}^i}{d{P}^i} \in L^{\infty}(\mathcal{F}^i_T)$, and such that 
$\varphi\in L^{1}(\Omega, \mathbf{F}_T,\mathbf{\widehat{P}}) $, $\mcY \subseteq L^{1 }(\Omega, \mathbf{F}_T,\mathbf{\widehat{P}})$ and all processes $X^{j}$, $j \in (i)$, are integrable under $\widehat{P}^i$.
From Theorem \ref{IFTAP:cone}, applied under such a $\widehat{\mathbf{P}}$, we infer the existence  of an equivalent collective martingale measure $\mathbf{Q}=(Q^1,\dots, Q^N)$ for the extended market, with bounded densities $\frac{dQ^i}{d\widehat{P}^i} \in L^{\infty}(\Omega, \mathcal{F}^i_T, \widehat P^i)$  for all  $i$. We see that  $\frac{dQ^i}{d{P}^i}=\frac{dQ^i}{d\widehat{P}^i} \frac{d\widehat P^i}{d{P}^i} \in L^{\infty}(\mathcal{F}^i_T)$ and, since $\varphi\in L^{1}(\Omega, \mathbf{F}_T,\mathbf{\widehat{P}}) $, $E_{Q^i}[\varphi^i]<+\infty$ for all $i$, we have $\mathbf{Q} \in \mathcal{M}^\varphi_e(\mcY)$.  Furthermore, since $Q^i$ is a martingale measure for the extended market of agent $i$,  $X^{J+i}_t = E_{Q^i}[X^{J+i}_T | \mathcal{F}^i_t]=E_{Q^i}[ f^i | \mathcal{F}^i_t]$. By definition
\[
\Pi^i_ f = X^{J+i}_0 = E_{Q^i}[X^{J+i}_T | \mathcal{F}_0] = E_{Q^i}[   f^i].
\]
We thus proved that 
\begin{align} 
\Pi(f) &\subseteq \left \{ (E_{Q^1}[   f^1],\dots, E_{Q^N}[   f^N]) \mid \mathbf{Q} \in \mathcal{M}^\varphi_e(\mcY) \text{ s.t. } \frac{dQ^i}{d{P}^i} \in L^{\infty}( \mathcal{F}^i_T) \text{ for all } i \right\}\\
&\subseteq \left \{ (E_{Q^1}[   f^1],\dots, E_{Q^N}[   f^N]) \mid \mathbf{Q} \in \mathcal{M}^\varphi_e(\mcY)\right\}
\end{align}
Conversely, take $x=(E_{Q^1}[   f^1],\dots, E_{Q^N}[   f^N]) \in \R^N$ for $ \mathbf{Q} \in \mathcal{M}^\varphi_e(\mcY)$. Define for $i=1,\dots,N$ the $Q^i-$ martingales $X^{J+i}_t = {E}_{Q^i}[   f^i | \mathcal{F}^i_t]$. Since $\mathbf{Q} \in \Me $, and $Q^i$ is an equivalent martingale measure for the $i$-extended market $(X^j)_{j\in(i)\cup\{J+i\}}$, $\mathbf{Q}$ acts as an equivalent collective martingale measure for the whole extended market. Thus, by Theorem \ref{IFTAP:cone},  ${\textbf{NCA}}(\mcY) $ holds true in the extended market.
Finally, $X^{J+i}_T = f^i$ for mere measurability and integrability arguments,
thus $x \in \Pi( f)$ by definition, so that $$\left \{ (E_{Q^1}[   f^1],\dots, E_{Q^N}[   f^N]) \mid \mathbf{Q} \in \mathcal{M}^\varphi_e(\mcY)\right\} \subseteq \Pi(f).$$ 
\end{proof}

%%%%%%%%%%%%%%%%%%%%%%%%%%%%%
%%%%%%%%%%%%%%%%%%%%%%%%%%%

\subsection{Properties of the set $\Pi(f)$ }\label{property:pi}

In this subsection, we assume that $\mcY$ is a finite-dimensional vector space containing $\mathbb{R}_0^N$. The motivation for imposing this assumption is discussed in Remark \ref{remrelopen}.\\
Example 3.3 and Proposition 3.6  of \cite{DFM25} demonstrate that, under $\mathbf{NCA}(\mcY)$  if $f\in L^0(\Omega, \mathbf{F}_T, P)$ is not $\mcY$-collectively replicable then:
\begin{enumerate}
    \item   $\Pi(f) \subseteq \mathbb R^N $, $N=2$, is not necessarily open;
    \item   $\Pi(f) \subseteq \mathbb R^N $, $N \geq 1$, is not a closed set. 
\end{enumerate}
 
If we assume that $\mcY \subseteq \mcY_0$ and $N=1$, then necessarily $\mcY=\{0\}$ and therefore $f$ is $\mcY$-collectively replicable if and only if it is classically replicable.
Thus the statement in Item 1 shows a difference with the classical case ($N=1$), where $\Pi(f) \subseteq \R$ is an open interval when $f$ is not replicable (Theorem 5.32 of \cite{FollmerSchied2}). We will show that the set $\Pi(f)$ is invariably \emph{relatively open}, regardless of whether $f$ is $\mcY$-replicable and for any $N \geq 1$. In the one-dimensional case ($N=1$), a relatively open nonempty interval is either open (which occurs when $f$ is not replicable) or a singleton (which occurs when $f$ is replicable).

This also shows that, in general, it is not possible to distinguish between replicable and non-replicable claims through the topological properties of the set $\Pi(f)$. Instead, the defining characteristic of a replicable claim is that $\Pi(f)$ comprises a single element, a point we elaborate on in the subsequent section.

\begin{definition} Let $A\subset \mathbb{R}^N$ be convex and denote by $$\operatorname{aff}(A):=\{ \lambda x_1+(1-\lambda) x_2 \mid x_1,x_2 \in A, \lambda \in \R\}$$ the affine hull generated by $A$. A convex set $A$ is relatively open if it is an open set in  $\operatorname{aff}(A)$, or equivalently (see \cite{Ro70} Theorem 6.4) if for every $x\in A$ and for every $x_1\in A$ there exists  $x_2\in A$ and a scalar $\lambda\in (0,1)$ such that
\begin{equation}\label{xxx}
    x = \lambda \,x_1 + (1-\lambda)\,x_2.
\end{equation}
\end{definition}
Notice that, for any $x \in \R^N$, $\operatorname{aff}(\{x\})=\{x\}$, so that any singleton is relatively open.\\

\begin{proposition}\label{propRO}
Suppose  that $\mcY$ is a finite dimensional vector space such that $\R^N_0\subseteq \mcY$ and that $\mathbf{NCA}(\mcY)$ hold true. Let $f\in L^0(\Omega,\textbf{F}_T,P)$ be given.
Then $\Pi(f) \subseteq \R^N$ is a relatively open nonempty convex set. 
\end{proposition}
\begin{proof}
Recall that we denote $L^{\infty}( \mathcal{F}^i_T):=L^{\infty}(\Omega, \mathcal{F}^i_T,P)$ and similarly for $L^{\infty}( \mathbf{F}_T)$. 
By \textbf{NCA}$(\mcY)$ and Theorem \ref{duality}, $M_e^\varphi(\mathcal{Y}) \neq \emptyset$ and thus $\Pi(f)$ is not empty. Convexity follows directly from the representation \eqref{44}.
First take in \eqref{45} ${P}^i=P$ for all $i$. Then
$$\Pi(f)  = \left\{ E_Q[f] := ( E_{Q^1}[f^1],\dots, E_{Q^N}[f^N]) \mbox{ }|\mbox{ } \mathbf Q \in \mathcal{M}_e^\varphi(\mathcal{Y}) \text{ and } \frac{dQ^i}{dP} \in L^\infty( F^i_T) \text{ for all } i \right\}.$$
Let $ x, x_1 \in \Pi(f) $ and let $ \mathbf Q_*, \mathbf Q_1 \in \mathcal{M}_e^\varphi(\mathcal{Y}) $,  with $\frac{dQ_*^i}{dP} \in L^\infty( F^i_T)$ and $\frac{dQ_1^i}{dP} \in L^\infty( F^i_T)$, such that
$x = E_{Q_*}[f]$, $x_1 = E_{Q_1}[f]$. We assume $x \neq x_1$, otherwise $x_2=x$ satisfies \eqref{xxx}.
Now in \eqref{45} we select $P^i= Q^i_*$ for all $i$, and write
$$\Pi(f)  = \left\{ E_Q[f] := ( E_{Q^1}[f^1],\dots, E_{Q^N}[f^N]) \mbox{ }|\mbox{ } \mathbf Q \in \mathcal{M}_e^\varphi(\mathcal{Y})  \text{  s.t. }\frac{dQ^i}{dQ_*^i} \in L^\infty(F^i_T) \text{ for all } i \right\}.$$
%Let $ x, x_1 \in \Pi(f) $ 
%If $p = p_1$ then $\check{p} = p_1$ is the vector we were looking for with $\lambda = \frac{1}{2}$.
%and let $ Q, Q_1 \in M_e^\varphi(\mathcal{Y}) $ such that
%$x = E_{Q}[f]$, $x_1 = E_{Q_1}[f]$.
Since \(x_1\in \Pi(f)\), by the representation above we may, possibly replacing the previously chosen
\(Q_1\), select \(Q_1\in \mathcal M^\varphi_e(\mathcal Y)\) such that
\[
x_1=E_{Q_1}[f]
\qquad\text{and}\qquad
\frac{dQ^i_1}{dQ^{i}_{\ast}}\in L^\infty(\mathcal F^i_T)
\quad\text{for every } i=1,\ldots,N .
\]
Consider the vectors of associated Radon-Nikodym derivatives w.r.t. $\mathbf{Q}_*$, namely $$
Z_1 = \left( \frac{dQ_1^1}{dQ_*^1}, \dots, \frac{dQ_1^N}{dQ_*^N} \right) \in L^{\infty}_+(\textbf{F}_T)  \mbox{ and } Z_* = \left( \frac{dQ_*^1}{dQ_*^1},  \dots, \frac{dQ_*^N}{dQ_*^N} \right) = (1, \dots, 1).$$ 
%Since both $Q$ and $Q_1$ belong to $M_e^\varphi(\mathcal{Y})$, we have $Q^i\sim Q^i_1 \sim P$ for every $i=1,\dots,N$, thus, given also that $\frac{dQ^i}{dP},\frac{dQ^i_1}{dP} \in L^\infty(\Omega, F^i_T,P)$, we have that $Z_1^i \in L^{\infty}_+(\Omega,F^i_T,Q^i)$ for every $i=1,\dots,N$.\\ 
%Since $p \neq p_1$, then there exists at least one $j \in \{1,\dots,N \}$ such that $Z^j \neq Z_1^j$. By applying the previous lemma, for such $j$, $\| Z_1^j\|_{\infty} > 1$, and, as a consequence, 
Since $\mathbf{Q}_1 \neq \mathbf{Q}_*$, $$ \alpha:= \max_{i=1,\dots, N} \| Z_1^i\|_{\infty} >  1 \text{ and } \beta: = \frac{1}{1+2( \alpha - 1)} \in (0,1)
$$ 
and notice that $\widehat{Z}: = \beta Z_1 + (1-\beta)Z_* \in  L^{\infty}_+(\textbf{F}_T) $ satisfies 
$$\|\widehat{Z}^i\|_{\infty}  \leq  \beta \|Z_1^i\|_{\infty}  +(1-\beta)\leq \frac{3}{2} \quad \forall i=1,\dots,N.$$
Define $Z_2:=2-\widehat{Z} \in L^{\infty}_+(\textbf{F}_T), $ as $ Z_2^i \geq 2 - \frac{3}{2} = \frac{1}{2}$, $ \forall i=1,\dots,N$.

Let $\widehat{\mathbf{Q}}$ be defined by $\widehat{Z}^i=\frac{d\widehat{Q}^i}{dQ_*^i}$ and $\mathbf{Q}_2$ be defined by $Z^i_2=\frac{dQ^i_2}{dQ_*^i} \in L^\infty(F^i_T)$, $i=1,\dots,N$. As both $ \mathbf{Q}_*, \mathbf{Q}_1 $  belong to the convex set $\mathcal{M}_e^\varphi(\mathcal{Y}) $ we get $\widehat{\mathbf{Q}} \in \mathcal{M}_e^\varphi(\mathcal{Y})$ and one easily checks that also $\mathbf{Q}_2 \in \mathcal{M}_e^\varphi(\mathcal{Y})$ and therefore $x_2: = E_{\mathbf{Q}_2}[f] \in \Pi(f)$. From the definitions of $\widehat{Z}$ and of $Z_2$ we get

$$Z_* =\frac{\beta Z_1}{1+\beta}+ \frac{Z_2}{1+\beta}$$
implying that 
$$x = E_{\mathbf{Q}_*}[f] =  \frac{\beta E_{Q_1}[f]}{1+\beta} + \frac{E_{Q_2}[f]}{1+\beta} = \frac{\beta }{1+\beta} x_1+ \frac{1}{1+\beta}x_2.$$

\end{proof}

\begin{remark}\label{comparison} Take $N=1$ and suppose that \textbf{NCA}$(\mcY)$ holds true.\\
    (i) If $f$ is classically replicable, then it is well known (Theorem 5.32, \cite{FollmerSchied2}) that $\Pi(f)$ is a singleton and thus it is relatively open, in accordance with Proposition \ref{propRO}. \\
    (ii) Observe that by \textbf{NCA}$(\mcY)$ and by the representation \eqref{44},  the set $\Pi(f)$ is a nonempty interval of $\R$ and by Proposition \ref{propRO} it is then a nonempty relatively open interval in $\R$.
    Any nonempty relatively open interval in $\R$ is either an open interval or a single point. In the latter case it is then also closed. Thus, if $f$ is not classically replicable, then $\Pi(f)$ is a relatively open nonempty interval that it is not closed (by Item 2 at the beginning of this section), thus it is an open interval (in accordance with Theorem 5.32 of \cite{FollmerSchied2}).
\end{remark}

\begin{remark}\label{remrelopen}
    The assertion that $\Pi(f)$ is relatively open under $\mathbf{NCA}(\mathcal Y)$ is not valid in full generality when $\mathcal Y$ is only a convex cone. This failure stems from the fact that the defining constraints of
$\mathcal M_e(\mathcal Y)$ become genuinely one-sided when $\mathcal Y$ is only a cone. 
%The obstruction is that the defining constraints of $\mathcal M_e(\mathcal Y)$ are (strictly) one-sided in the case of a cone.  
Example \ref{example42} shows that in the case of a finitely generated convex cone $\Me$ can be a closed set and $\Pi(f)$ is not relatively open.
\end{remark}

\subsection{Strongly collectively replicable claim}\label{sec52}

Suppose that Assumption \ref{ass:cone} holds, that $f$ is a collectively $\mcY$-replicable bounded claim and \textbf{NCA}$(\mcY)$ holds.   
 If $(-f)$ is $\mcY$-replicable and \textbf{NCA}$(-\mcY)$ holds, or if $\mcY$ is a vector space,  
 %Proposition \ref{uguali} 
  Corollary \ref{corexp0} implies that the quantity $\sum_{i=1}^N E_{Q^i}[f^i]$
is the same for every $\mathbf{Q}=(Q^1,\dots,Q^N)\in\Me$. 

In other words, the aggregate valuation of the claim $f$ is invariant across all pricing vectors in $\Me$.  
Nevertheless, except in the classical case $N=1$, this invariance does not imply that the pricing set $\Pi(f)$ consists of a single element, even when $\mcY$ is a vector space.

The next example, which describes a collectively incomplete market with a vector space $\mcY$ of possible exchanges, shows that for a collectively $\mcY$-replicable claim $f$, the set $\Pi(f)$ may fail to be a singleton.

\begin{example}\label{ex1}
    Consider the example 7.2.3 \cite{BDFFM25} with $|\Omega|=6$, $\mathcal T:=\{0,1,2\}$, $N=2$ agents, each investing in one single stock and with $\mathcal{Y} = \left\{ Y \in (L^0(\Omega, \mathbf{F}_1, P))^2 \mid  Y^1+Y^2 = 0  \right\}$ the vector space of allowed exchanges. It was shown that \textbf{NCA}$(\mcY)$ holds. Take as contingent claims $f^1=(1,1,0,0,0,0)$ and $f^2=(0,0,0,0,-1,-1)$, $f=(f^1,f^2)$. Using the set of collective martingale measures on page 49 \cite{BDFFM25},  we compute $E_{Q^1}[f^1]=(1/2)q$, $E_{Q^2}[f^2]=-(1/2)q$, with  $0< q< 1$, thus $E_{Q^1}[f^1]+E_{Q^2}[f^2]=0$ for all $\mathbf Q \in \Me$ and $f$ is $\mcY$-collectively replicable by Proposition \ref{uguali}. Moreover, using Theorem \ref{BaseFTAPII} $$\Pi(f)=\bigg\{\bigg(\frac{1}{2}q,-\frac{1}{2}q\bigg) \mid 0<q<1 \bigg\} \subset \mathbb R^2$$ is not a singleton.
\end{example}

This is a departure from the classical theory. Indeed, a single classically replicable contingent claim $f^i$ can be uniquely written as $f^i=x^i+k^i $, $x^i \in \mathbb R, k^i \in K_i, \, $even if $k^i $ may be obtained via multiple strategies. For \emph{any} equivalent martingale measure $Q \in M^{i,\varphi^i}_e=\{Q^i \in M^{i}_e \text{ s.t. } \varphi^i\in L^1(\Omega,\mathcal{F}^i_T,Q^i) \}$, one has $E_{Q}[k^i]=0$ and thus in the classical case $E_{Q}[f^i]=x^i+E_{Q}[k^i]=x^i$ is a constant value, namely the unique price of $f^i$. While, even when $\mcY$ is a vector space,  a collectively $\mcY$-replicable claim $f=( f^1,\dots,  f^N) \in L^{0}(\Omega, \mathbf{F}_T, P)$   can be obtained in different ways, namely it can be written as $f=x_1+k_1+Y_1=x_2+k_2+Y_2$ for the elements belonging to the obvious sets. For $Q_1 \in \mathcal{M}^\varphi_e(\mcY)$ and $Q_2 \in \mathcal{M}^\varphi_e(\mcY)$, it is still true that $E_{Q_1^i}[k_1^i]=0$ and $E_{Q_2^i}[k_2^i]=0$, but it is possible that, for some $i$, $E_{Q_1^i}[Y_1^i] \neq E_{Q_2^i}[Y_2^i]$ for $Y_1 \in \mcY$ and $Y_2 \in \mcY$, despite $\sum_{i=1}^N E_{Q_1^i}[Y_1^i]=\sum_{i=1}^N E_{Q_2^i}[Y_2^i]=0$. Thus it may happen, as in Example \ref{ex1}, that $$E_{Q_1^i}[f^i]=x_1^i+E_{Q_1^i}[Y_1^i] \neq x_2^i+E_{Q_2^i}[Y_2^i]=E_{Q_2^i}[f^i], $$  so that $\Pi( f)$ may not be reduced to a singleton.
An excessively large set of possible exchanges, $\mcY$, leads to this phenomenon. We will show (see Proposition \ref{propUnique}) that by applying the restricted set of exchanges $\widehat{\mcY}$ this issue is resolved.

\begin{definition}\label{defYrestricted} Let $\mcY \subseteq L^{0}(\Omega, \mathbf{F}_T, P)$ be a convex cone and consider the convex cone of \emph{restricted} exchanges  defined by
    \begin{equation}
    \label{hatYdef}
        \widehat{\mcY}:=\left \{ Y \in  \mcY \mid E_{Q^i}[Y^i]=0 \,\, \forall \, i=1,\dots,N \text{ and } \forall \, Q \in \Me \right \} \subseteq \mcY.
    \end{equation}
\end{definition}

\noindent Observe that $0 \in \widehat{\mcY}$, so that $\widehat{\mcY} \neq\emptyset$. Similarly to Definition \ref{defrep}, we consider

\begin{definition}
 The vector of contingent claims $ f=( f^1,\dots,  f^N) \in L^{0}(\Omega, \mathbf{F}_T, P)$    is $\widehat{\mcY}$-collectively replicable (or strongly collectively replicable) if $ f\in \R^N+ {\sf X}_{i=1} ^{N} K_i+ \widehat{\mcY}  $.
\end{definition}
Since $\{0\}\subseteq \widehat{\mcY} \subseteq \mcY$, obviously classical replicability (when no exchanges are admitted) is a stronger requirement than $\widehat{\mcY}$-replicability, which is a stronger requirement than $\mcY$-replicability and thus
\begin{equation}\label{eqRhoRho}
\rho^{{N}}_+\geq \rho^{\widehat{\mcY}}_+\geq \rho^{{\mcY}}_+. 
\end{equation}
As shown in the Example \ref{example22} in Section \ref{example1}, strict inequalities in \eqref{eqRhoRho} may hold, even when $\mcY$ is a vector space.
Observe also that $\mathbb R ^N_0 $ is never included in $\widehat{\mcY} $, except if $N=1$ where $\mathbb R ^N_0=\{0\} $ and except when $\Me = \emptyset$. Additionally, it always holds
\begin{align}
    \label{eqalityplusrn0}
    \R^N+ {\sf X}_{i=1} ^{N} K_i+ \widehat{\mcY}=&\R^N+ {\sf X}_{i=1} ^{N} K_i+ \big(\widehat{\mcY}+\R^N_0\big), \\
    \R^N+ {\sf X}_{i=1} ^{N} K_i- \widehat{\mcY}=&\R^N+ {\sf X}_{i=1} ^{N} K_i- \big(\widehat{\mcY}+\R^N_0\big)
\end{align}
which implies that $\pm f$    is $\widehat{\mcY}$-collectively replicable if and only if it is $(\widehat{\mcY}+\R^N_0)$-collectively replicable.

Moreover, we know from \cite{BDFFM25} Lemma 4.4 Item 6 that
\begin{equation}\label{rhorho22}
\rho^{(\widehat{\mcY}+\R^N_0)}_+=\rho^{\widehat{\mcY}}_+.   
\end{equation}
We recall that by definition
\begin{equation}\label{MartingaleMeasuresBis}
\begin{split}
\mathcal{M}_e(\widehat{\mcY})=\bigg\{ \mathbf{Q}=(Q^1,\dots,Q^N) \in {\sf X}_{i=1} ^{N}\mie   \mid  \widehat{\mcY} \subseteq L^{1 }(\Omega, \mathbf{F}_T,Q)  \text { and  } \sum_{i=1}^N E_{Q^i}[Y^i] \leq 0 \;\forall\, Y \in \widehat{\mcY}  \bigg\}
\end{split}
\end{equation}
and we observe that  $\mathcal{M}_e(\widehat{\mcY}+\R^N_0)=\mathcal{M}_e(\widehat{\mcY}) \supseteq \mathcal{M}_e({\mcY})$.\\
Suppose that $\mcY$ satisfies Assumption \ref{ass:cone}. Then $\widehat{\mcY}+\R^N_0$ is a convex cone containing $\R^N_0$ and, since $(\widehat{\mcY}+\R^N_0) \subseteq \mcY$, then \textbf{NCA}$({\mcY})$ implies  
\textbf{NCA}$(\widehat{\mcY}+\R^N_0)$. In order to apply previous results, we show that  Assumption \ref{ass:cone} holds for the convex cone $\widehat{\mcY}+\R^N_0$, by proving that the latter is finitely generated.
  \begin{lemma}
  \label{lemma:hatyfinitegen}
     Suppose $\mcY$ is a finitely generated convex cone. Then  $\widehat{\mcY}$ given in \eqref{hatYdef} is a finitely generated convex cone too. In particular, if $\mcY$ satisfies Assumption \ref{ass:cone}, so does  $\widehat{\mcY}+\R^N_0$.
 \end{lemma}
 \begin{proof}

Since \(\mathcal Y\) is finitely generated, \(
\mathcal Y=\mathrm{cone}(\{Y_1,\ldots,Y_R\})
\) for some
\(Y_1,\ldots,Y_R\in L^0(\Omega, \mathbf{F}_T, P) \).
Let
\(
V:=\mathrm{span}(\{Y_1,\ldots,Y_R\}).
\)
Then \(V\) is a finite-dimensional vector space and
$
\mathcal Y\subseteq V.
$
For every \(\mathbf Q=(Q^1,\ldots,Q^N)\in \mathcal M_e(\mathcal Y)\) and every
\(i=1,\ldots,N\), define the linear functional
$\ell_{Q,i}:V\to \mathbb R,
\ell_{Q,i}(Y):=E_{Q^i}[Y^i].$
Then, by definition,
\[
\widehat{\mathcal Y}
=
\mathcal Y\cap 
\bigcap_{\mathbf Q\in \mathcal M_e(\mathcal Y)}
\bigcap_{i=1}^N
\ker(\ell_{Q,i}).
\]
Since \(
\bigcap_{\mathbf Q\in \mathcal M_e(\mathcal Y)}
\bigcap_{i=1}^N
\ker(\ell_{Q,i})\) is an intersection of linear subspaces of the finite-dimensional
space \(V\), it is itself a finite dimensional linear subspace of \(V\). In particular, it is a finitely generated convex cone and by the Minkowski--Weyl Theorem (\cite{Ro70} Theorem 19.1), it is a polyhedral cone.
At the same time \(\mathcal Y\) too is a finitely generated cone in the finite-dimensional
space \(V\), hence a polyhedral cone by the same argument. The intersection of two polyhedral cones is still polyhedral by definition, so
$\widehat{\mathcal Y}$
is again a polyhedral cone.
By \cite{Ro70} Theorem 19.1  every polyhedral cone in a finite-dimensional
space is finitely generated. Therefore \(\widehat{\mathcal Y}\) is finitely
generated.
\end{proof}

We are now ready to show that using the restricted set $\widehat{\mcY}$ of possible exchanges the set $\Pi (f)$, of \textbf{NCA}$(\mcY)$-price of a 
$\widehat{\mcY}$-collectively replicable claim $f$, is a singleton. Example \ref{example42} demonstrates that the assumption $\mathbf{NCA}(-\mcY)$ cannot be dispensed with.

\begin{proposition}\label{propUnique}
  Suppose  Assumption \ref{ass:cone}, $\mathbf{NCA}(\mcY)$ and $\mathbf{NCA}(-\mcY)$ hold true and let $ f \in L^{0}(\Omega, \mathbf{F}_T, P)$. Then the following are equivalent:
  \begin{enumerate}
      \item Both $ f $ and $(-f)$    are $\widehat{\mcY}$-collectively replicable;
    \item   $\Pi(f)$ is a singleton;
      \item  $\rho^{\widehat{\mcY}}_+(f)=\rho^{\widehat{\mcY}}_-(f).$    
  \end{enumerate}
\end{proposition}
\begin{proof}
($1.\Rightarrow 2.$).  Let $ f\in L^{0}(\Omega, \mathbf{F}_T,P)$ be $\widehat{\mcY}$-collectively replicable. Then we have: $f=x+k+Y$ with $x \in \mathbb R^N$, $k \in {\sf X}_{i=1} ^{N} K_i$ and $Y \in \widehat{\mcY}$. Take any $\mathbf{Q} \in \mathcal{M}^\varphi_e(\mcY)$ and compute $E_{Q^i}[f^i]=x^i+E_{Q^i}[k^i]+E_{Q^i}[Y^i]=x^i,$ by definition of $\widehat{\mcY}$. Then from Theorem \ref{BaseFTAPII} $\Pi(f)$ is a singleton.

($1.\Leftarrow 2.$): If $\Pi(f)$ is a singleton, by Theorem \ref{BaseFTAPII} we have for some $\varphi^i$ s.t. $\max \{|f^i|,|Y^i_1|,\ldots, |Y^i_R|\}\leq \varphi^i$ and some 
 $\widehat {\mathbf{Q}}\in \mathcal{M}^\varphi_e(\mcY)$. 
$$\Pi( f) = \left\{ (E_{Q^1}[   f^1],\dots, E_{Q^N}[   f^N]) \mid \mathbf{Q} \in \mathcal{M}^\varphi_e(\mcY)\right\} = \{(E_{\widehat Q^1}[   f^1],\dots, E_{\widehat  Q^N}[   f^N])\}.$$
Note that consequently $\Pi(-f) =  \{(E_{\widehat Q^1}[-f^1],\dots, E_{\widehat  Q^N}[-f^N])\}$.
Using \eqref{rem:rhopm} and \eqref{pricing:hedging}  we have  $
  \rho^{\mcY}_+(f)=\sum_{i=1}^N E_{\widehat Q^i}[f^i]=\rho^{\mcY}_-(f)$. Then by Proposition \ref{diversi} Item 1, both $f$ and  $(- f) $ are $\mcY$-collectively replicable.

Hence, 
\begin{equation}\label{Eqfxky}
 f=x+k+Y=\tilde x+\tilde k -\tilde Y  
\end{equation}
 with $x,\tilde x \in \mathbb R^N$, $k,\tilde k \in {\sf X}_{i=1} ^{N} K_i$ and $Y,\tilde Y \in \mcY$. We now show that actually we can take $Y\in\widehat{\mcY}$ in the previous expression for $f$. Take $\mathbf{Q} \in \mathcal{M}^\varphi_e(\mcY)$ and  $\mathbf{R} \in \mathcal{M}^\varphi_e(\mcY)$ and compute $E_{Q^i}[f^i]=x^i+E_{Q^i}[k^i]+E_{Q^i}[Y^i]=x^i+E_{Q^i}[Y^i]$ and $E_{R^i}[f^i]=x^i+E_{R^i}[k^i]+E_{R^i}[Y^i]=x^i+E_{R^i}[Y^i]$. Since $\Pi(f)$ is a singleton, we must have $E_{Q^i}[f^i]=E_{R^i}[f^i]$ so that $E_{Q^i}[Y^i]=E_{R^i}[Y^i]$. Consequently, $E_{Q^i}[Y^i]:=\widehat x^i$ is independent from $\mathbf{Q} \in \mathcal{M}^\varphi_e(\mcY)$. 
 % We can then write $f^i=x^i+k^i+Y^i=(x^i+\widehat x^i)+k^i+(Y^i-\widehat x^i)$. 
  From \eqref{Eqfxky}, using Corollary \ref{corexp0}, $0=\sum_{i=1}^N E_{Q_i}[Y^i]=\sum_{i=1}^N \widehat{x}^i$.
  Thus $\widehat{x} \in \R^N_0$ and
  $Y+(-\widehat{x})\in \mcY+\R^N_0=\mcY$. 
  It remains to prove that
$E_{Q^i}[Y^i-\widehat{x}^i]=0
\, \text{ for every } i \text{ and every } \mathbf Q\in\mathcal M_e(\mathcal Y),
$
since, at this stage, we only know that this equality holds for every
$
\mathbf Q\in\mathcal M_e^\varphi(\mathcal Y).
$
Once this is established, it follows that
$
f=x+k+Y=(x+\widehat{x})+k+\widehat{Y},
$
where
$
\widehat{Y}:=Y-\widehat{x}\in\widehat{\mathcal Y}.
$
Therefore, \(f\) is \(\widehat{\mathcal Y}\)-collectively replicable.
  
  We now prove that $E_{Q^i}[Y^i-\widehat{x}^i]=0$ for every $i$ and $\mathbf{Q}\in \mathcal{M}_e(\mcY)$. 
  Let us fix $j\in\{1,\dots,N\}$ and define $g^j:=Y^j-\widehat{x}^j, g^i=0$ for $i\neq j$, $i=1,\dots, N$. 
  Set
$
\varphi^{\prime i}
:=
\varphi^i+|Y^i-\widehat x^i|,
\, i=1,\ldots,N .
$
Then \(\varphi'\) still dominates \(|f|\) and the generators of \(\mathcal Y\), and it also
dominates \(|Y-\widehat x|\). Moreover,
\(
\mathcal M^{\varphi'}_e(\mathcal Y)
\subseteq
\mathcal M^\varphi_e(\mathcal Y),
\)
and hence, by the preceding argument,
$
E_{Q^i}[Y^i-\widehat x^i]=0$
for every $i=1,\ldots,N$
 and every $\mathbf{Q}\in\mathcal M^{\varphi'}_e(\mathcal Y)$.  
  By Theorem
  \ref{duality} and \eqref{rem:rhopm} we see that 
  \begin{align*}
      \rho^{{\mcY}}_+(g^j)&=\sup \left \{\sum_{i=1}^N E_{Q^i}[g^i]\mid \mathbf{Q}\in \mathcal{M}^{{\varphi'}}_e(\mcY) \right \}=\sup \left \{E_{Q^j}[Y^j-\widehat{x}^j]\mid \mathbf{Q}\in \mathcal{M}^{{\varphi'}}_e(\mcY) \right \}=0\\
      \rho^{{\mcY}}_-(g^j)&=\inf \left \{\sum_{i=1}^N E_{Q^i}[g^i]\mid \mathbf{Q}\in \mathcal{M}^{{\varphi'}}_e(\mcY) \right \}=\inf\left \{E_{Q^j}[Y^j-\widehat{x}^j]\mid \mathbf{Q}\in \mathcal{M}^{{\varphi'}}_e(\mcY) \right \}=0.
  \end{align*}
  Take now $\widetilde{\varphi}^i=|Y^i-\widehat{x}^i|+\max\{\abs{Y^i_1},\dots \abs{Y^i_R}\}$ so that $\widetilde{\varphi}^i\geq\max \{|g^i|,|Y^i_1|,\ldots, |Y^i_R|\}$ for every $i=1,\dots,N$. Note that $M_e^{\widetilde{\varphi}}(\mcY)=M_e(\mcY)$. Again, by Theorem
  \ref{duality} and \eqref{rem:rhopm}
  \begin{align*}
      0&=\rho^{{\mcY}}_+(g^j)=\sup \left \{\sum_{i=1}^N E_{Q^i}[g^i]\mid \mathbf{Q}\in \mathcal{M}^{\widetilde{\varphi}}_e(\mcY) \right \}=\sup \left \{E_{Q^j}[Y^j-\widehat{x}^j]\mid \mathbf{Q}\in \mathcal{M}^{\widetilde{\varphi}}_e(\mcY) \right \}\\
      &\geq \inf\left \{E_{Q^j}[Y^j-\widehat{x}^j]\mid \mathbf{Q}\in \mathcal{M}^{\widetilde{\varphi}}_e(\mcY) \right \}=\inf \left \{\sum_{i=1}^N E_{Q^i}[g^i]\mid \mathbf{Q}\in \mathcal{M}^{\widetilde{\varphi}}_e(\mcY) \right \}=\rho^{{\mcY}}_-(g^j)=0.
  \end{align*}

  Then $$0=\inf\left \{E_{Q^j}[Y^j-\widehat{x}^j]\mid \mathbf{Q}\in \mathcal{M}^{\widetilde{\varphi}}_e(\mcY) \right \}=\sup\left \{E_{Q^j}[Y^j-\widehat{x}^j]\mid \mathbf{Q}\in \mathcal{M}^{\widetilde{\varphi}}_e(\mcY) \right \}=0$$
  which yields $E_{Q^j}[Y^j-\widehat{x}^j]=0$ for every $\mathbf{Q}\in \mathcal{M}^{\widetilde{\varphi}}_e(\mcY) =\mathcal{M}_e(\mcY) $, $j=1,\dots,N$.

 We finally show that also $-f$ is $\widehat{\mcY}$-collectively replicable. Indeed, 
 $$-f=-\tilde x -\tilde k+\tilde{Y}$$
 and $$\tilde Y=-(x+k+Y)+(\tilde x+\tilde k )=-(x+\widehat x-\tilde x)-(k-\tilde{k})-\widehat Y.$$
 Hence $$\tilde{Y}+(x+\widehat x-\tilde x)=-(k-\tilde{k})-\widehat Y.$$
 From Corollary \ref{corexp0} and since $\widehat x \in \R^N_0$, we deduce that $(x+\widehat x-\tilde x)\in \R^N_0$. Hence $\tilde{Y}+(x+\widehat x-\tilde x)\in \mcY$.
Furthermore, since $\tilde Y, \widehat{Y} \in \mcY$ and $\mcY \in L^1(\Omega,\mathbf{F}_T,\mathbf{Q})$ for every $\mathbf{Q}\in \mathcal M_e(\mcY)$,  the latter equality shows that $(k-\tilde{k})\in L^1(\Omega,\mathbf{F}_T,\mathbf{Q})$ for all such $\mathbf Q$. Since $Q^i\in M^i_e$, we have $E_{Q^i}[k^i-\widetilde{k}^i]=0$. As also $E_{Q^i}[\widehat{Y}^i]=0$ for every $\mathbf{Q}\in \mathcal M_e(\mcY), \, i=1,\dots, N$ we get $E_{Q^i}[\widetilde{Y}^i+ x^i+\widehat x^i-\tilde x^i]=0$ for every $i=1,\dots, N$, $\mathbf{Q}\in \mathcal M_e(\mcY)$. Thus, $\tilde{Y}+(x+\widehat x-\tilde x)\in \widehat \mcY$.
Finally we can write
$$-f=-\Big(\tilde x+(x+\widehat x-\tilde x)\Big)  -\tilde k+\Big(\tilde{Y}+(x+\widehat x-\tilde x)\Big),$$
i.e. $-f$ is $\widehat{\mcY}$-collectively replicable.

 ($1.\Leftrightarrow 3.$): 
 From Lemma \ref{lemma:hatyfinitegen} we know that Assumption \ref{ass:cone} holds replacing $\mcY$ with $(\widehat{\mcY}+\R^N_0)$. As $\mathbf{NCA}(\widehat{\mcY}+\R^N_0)$ and $\mathbf{NCA}(-(\widehat{\mcY}+\R^N_0))$ hold true, we apply Proposition \ref{diversi} Item 2. to infer that 
  $${\pm} f\in \R^N+ {\sf X}_{i=1} ^{N} K_i+ \big(\widehat{\mcY}+\R^N_0\big)\quad \Longleftrightarrow \quad\rho^{(\widehat{\mcY}+\R^N_0)}_-(f)=\rho^{(\widehat{\mcY}+\R^N_0)}_+(f)$$
   The desired  equivalence then follows from $\rho^{(\widehat{\mcY}+\R^N_0)}_\pm(f)=\rho^{\widehat{\mcY}}_\pm(f)$, see \eqref{rhorho22} together with \eqref{eqalityplusrn0}.
\end{proof}

From Theorem \ref{duality} we deduce also the following result.
\begin{corollary}[Pricing-hedging duality for $\widehat{\mcY}$]\label{Propduality:R} 
Suppose  Assumption \ref{ass:cone} and $\mathbf{NCA}(\mcY)$ hold true and let $\mcY=\co(Y_1,\dots,Y_R)$.
Then, for any $f\in L^{0}(\Omega, \mathbf{F}_T,P)$ and for any $\varphi\in L^{0}(\Omega, \mathbf{F}_T,P)$  such that $\varphi^i$
dominates $|f^i|$, \(1\) and the absolute values of every $i$-component of every generator of $\widehat \mcY$ for every $i$,
we have $\mathcal{M}^{\varphi}_e(\widehat{\mcY}) \neq \emptyset$ and
\begin{equation}\label{pricing:hedgingBis}
  \rho^{\widehat{\mcY}}_+(f)=\sup \left \{\sum_{i=1}^N E_{Q^i}[f^i]\mid \mathbf{Q}\in \mathcal{M}^{\varphi}_e(\widehat{\mcY}) \right \}>-\infty. 
\end{equation}
  Moreover, when finite, $ \rho^{\widehat{\mcY}}_+(f) $ is attained for some $m \in \mathbb R^N.$
\end{corollary}

\begin{proof}
We know (Lemma \ref{lemma:hatyfinitegen}) that Assumption \ref{ass:cone} holds replacing $\mcY$ with $(\widehat{\mcY}+\R^N_0)$. By Theorem \ref{duality}, $\mathcal{M}^{\varphi}_e(\widehat{\mcY}+\R^N_0) \neq \emptyset$ so that $\mathcal{M}^{\varphi}_e(\widehat{\mcY}) \neq \emptyset$.
    Applying again Theorem \ref{duality} and \eqref{pricing:hedging}, replacing $\mcY$ with $(\widehat{\mcY}+\R^N_0)$, we thus get 
  \begin{align*}\label{pricing:hedgingbis}
  \rho^{\widehat{\mcY}}_+(f)&=\rho^{(\widehat{\mcY}+\R^N_0)}_+(f)=\sup \left \{\sum_{i=1}^N E_{Q^i}[f^i]\mid \mathbf{Q}\in \mathcal{M}^{\varphi}_e(\widehat{\mcY}+\R^N_0) \right \}\\
  &=\sup \left \{\sum_{i=1}^N E_{Q^i}[f^i]\mid \mathbf{Q}\in \mathcal{M}^{\varphi}_e(\widehat{\mcY}) \right \}.
  %\geq \sup \left \{\sum_{i=1}^N E_{Q^i}[f^i]\mid Q\in \mathcal{M}^{\varphi}_e(\mcY) \right \}=\rho^{\mcY}_+(f).
\end{align*}
   and the attainment follows again from Theorem \ref{duality}.
\end{proof}

%%%%%%%%%%%%%%%%%%%%%%%%%%%
\subsection{Collectively complete markets}

The following version of the Second Fundamental Theorem of Asset Pricing (CFTAP II) was proved in \cite{DFM25}, under the assumption of $\mcY$ being a vector space.

\begin{theorem}[CFTAP II, \cite{DFM25} Theorem 2.10]\label{completeTH} 
Assume $\mcY$ is a finite dimensional vector space containing $\R^N_0$ and for which $\mathbf{NCA}(\mcY)$ holds.  
The market is collectively complete if and only if  $\Me$  is a singleton, namely $$\text{every } f\in L^{0}(\Omega, \mathbf{F}_T,P) \text{ is }  \mcY\text{-collectively replicable}  \Longleftrightarrow \Me \text{ is a singleton}.$$
\end{theorem}

When \(\mathcal Y\) is a convex cone satisfying Assumption \ref{ass:cone}, rather than a vector space, Proposition \ref{propUnique} shows that the condition \(\mathbf{NCA}(\mathcal Y)\) alone is not sufficient to obtain the desired results, and must be complemented by \(\mathbf{NCA}(-\mathcal Y)\).
Clearly,
\[
\mathbf{NCA}(\mathrm{span}(\mathcal Y))
\;\Longrightarrow\;
\mathbf{NCA}(\mathcal Y)
\text{ and }
\mathbf{NCA}(-\mathcal Y).
\]
However, the converse implication is false in general, as illustrated by the counterexample in Section \ref{example1}, Example \ref{example43}.
Under the completeness assumption on the market, however, the two conditions
$
\mathbf{NCA}(\mathcal Y)
\, \text{ and } \,
\mathbf{NCA}(-\mathcal Y)
$
are together equivalent to
$
\mathbf{NCA}(\mathrm{span}(\mathcal Y))
$.

\begin{lemma}\label{lemma:sticazzi} Suppose that Assumption \ref{ass:cone} holds and $L^0(\Omega, \mathbf{F}_T,P)=\R^N+{\sf X}_{i=1}^N K_i +\mcY$. Then  
\[\mathbf{NCA}(\mathcal Y)\text{ and } \mathbf{NCA}(\mathcal -\mcY) \Longleftrightarrow \mathbf{NCA}(\mathrm{span}(\mathcal Y))\] 
\end{lemma}

\begin{proof}
Assume $\mathbf{NCA}(\mathcal Y)\text{ and } \mathbf{NCA}(\mathcal -\mcY)$. By Theorem \ref{IFTAP:cone}, \(\mathbf{NCA}(\mcY)\) implies the existence of
$
\mathbf Q=(Q_1,\dots,Q_N)\in \mathcal M_e(\mcY)
$.
 Take any \(Y\in \mcY \subseteq L^{1}(\Omega, \mathbf{F}_T,\mathbf Q)\). By $\mcY$-collective completeness, $f=Y$ and $(-Y)$ are $\mcY$-replicable and thus 
by Corollary \ref{corexp0} Item 2 we get $\sum_{i=1}^N E_{Q_i}[Y^i]=0$. As this is true for any $Y \in \mcY$, by linearity
\[
\sum_{i=1}^N E_{Q_i}[Z^i]=0
\qquad\forall\,Z\in \mathrm{span}(\mcY),
\]
which ensures
$
\mathbf Q\in \mathcal M_e(\mathrm{span}(Y)).
$
The space $\mathrm{span}(Y)$ is itself a finitely generated convex cone containing $\R^N_0$, so Assumption \ref{ass:cone} is met also for $\mathrm{span}(Y)$. Applying again Theorem \ref{IFTAP:cone}, we infer that
$\mathbf{NCA}(\mathrm{span}(Y))$ holds.
   
\end{proof}

The following theorem strengthens Theorem \ref{completeTH} in two respects. First, it applies when $\mcY$ is merely a convex cone, rather than a vector space. Second, it provides a characterization of both collective completeness and \emph{strong collective completeness} in terms of the uniqueness of collective martingale measures.
Example \ref{example42} exhibits a market that is $\mcY$-collectively complete and satisfies $\mathbf{NCA}(\mcY)$, while
$
|\Me|=\infty
$
and
$
\rho_-^{\mcY}(f)<\rho_+^{\mcY}(f)
$
for some claim \(f\in L^{0}(\Omega,\mathbf{F}_T,P)\). Therefore, the assumption $\mathbf{NCA}(-\mcY)$ is essential.

\begin{theorem}[CFTAP II for a cone $\mcY$]\label{completeTHbis}Suppose  Assumption \ref{ass:cone}, $\mathbf{NCA}(\mcY)$ and $\mathbf{NCA}(-\mcY)$ hold true. Assume\footnote{This assumption simplifies the statement, but could be removed. If this is not assumed, the Theorem remains true if Item 2 and 3 are replaced by the condition: the set of restrictions $(Q^i|_{\mathcal{F}^i_T})_i$, as $\mathbf{Q}$ ranges in $\mathcal{M}_e(\mcY)$ (resp. $\mathcal{M}_e(\widehat{\mcY})$) is a singleton.}  that $\mathcal{F}^i_T=\mathcal{F}$ for every $i=1,\dots,N$.
Then the following are equivalent:
\begin{enumerate}
    \item The market is $\mcY$-collectively complete, namely every  $f\in L^{0}(\Omega, \mathbf{F}_T,P)$ is  $\mcY$-collectively replicable;
    %\item $L^{0}(\Omega, \mathbf{F}_T,P)=\R^N + {\sf X}_{i=1} ^{N} K_i+\mathcal Y = \R^N +{\sf X}_{i=1} ^{N} K_i+\text{span}(\mcY)$.
    \item $\Me$  is a singleton;
    \item  $\rho^{ \mcY}_+ =\rho^{ \mcY}_-$ on $L^{0}(\Omega, \mathbf{F}_T,P)$;
    \item The market is $\widehat{\mcY}$-collectively complete, namely every  $f\in L^{0}(\Omega, \mathbf{F}_T,P)$ is  $\widehat{\mcY}$-collectively replicable;
    \item $\mathcal{M}_e(\widehat{\mcY})$ is a singleton;
    \item $\rho^{ \widehat{\mcY}}_+=\rho^{ \widehat{\mcY}}_-$  on $L^{0}(\Omega, \mathbf{F}_T,P)$;
\end{enumerate}
    In any such case, $\Me=\mathcal{M}_e(\widehat{\mcY})$, $\rho^{ \widehat{\mcY}}_+=\rho^{ \mcY}_+$, $\rho^{ \widehat{\mcY}}_-=\rho^{ \mcY}_-$ and $\mathbf{NCA}(\mathrm{span}(\mathcal Y))$ holds.
\end{theorem}
\begin{proof}
We know that  $\mathcal{M}_e(\widehat{\mcY}+\R^N_0)=\mathcal{M}_e(\widehat{\mcY}) \supseteq \mathcal{M}_e({\mcY})\neq\emptyset$  by \textbf{NCA}$(\mcY)$. Recall that we indicate with $\{e_i\}_{i=1,\ldots,N}$ the canonical basis of $\R^N$.

$ (1. \Rightarrow 2.)$: Pick $A\in\mcF$. For every $j=1,\dots,N$ we have that both $\pm 1_{A}e_j$ are $\mcY$-collectively replicable. 
Pick $\mathbf{Q}, \widehat{\mathbf{Q}}\in \Me$. From Proposition \ref{diversi} Item 2 we have $\rho^{ \mcY}_-( 1_Ae_j)=\rho^{ \mcY}_+( 1_Ae_j)$. At the same time, choosing $\varphi^i=1+\max \{|Y^i_1|,\ldots, |Y^i_R|\}$, we have $\mathbf{Q}, \widehat{\mathbf{Q}}\in \mathcal{M}^{\varphi}_e(\mcY)=\Me$.  Theorem \ref{duality} allows us to write
    $$\rho^{ \mcY}_-( 1_Ae_j)\leq E_{Q^j}[1_A]=\sum_{i=1}^NE_{Q^i}[1_Ae^i_j]\leq \rho^{ \mcY}_+( 1_Ae_j)$$ and the same holds for $E_{\widehat{Q}^j}[1_A]$, so that $E_{Q^j}[1_A]=E_{\widehat{Q}^j}[1_A]$. Since this can be repeated for every $j$ and $A\in\mcF$, it follows that $\Me$ (which is nonempty) also consists of at most one element.

$ (2. \Rightarrow 1.)$:  Pick $f\in L^{0}(\Omega, \mathbf{F}_T,P)$ and $\varphi\in L^{0}(\Omega, \mathbf{F}_T,P)$ with
$\max \{|f^i|,|Y^i_1|,\ldots, |Y^i_R|\}\leq \varphi^i$. By Theorem \ref{duality}, $ \mathcal{M}^\varphi_e(\mcY)\neq \emptyset$ and at the same time $ \mathcal{M}^\varphi_e(\mcY)\subseteq \Me$, the latter being a singleton, $\mathcal{M}^\varphi_e(\mcY)$ is itself a singleton, say $\{\widehat{\mathbf{Q}}\}$. Since $\rho^{\mcY}_-(f)=-\rho^{\mcY}_+(-f)$, \eqref{pricing:hedging} implies $$\rho^{\mcY}_+(f)=\sum_{i=1}^N E_{\widehat{Q}^i}[f^i]=\rho^{\mcY}_-(f).$$
Then we can invoke Proposition \ref{diversi} Item 1 which yields in particular 
$f\in \R^N + {\sf X}_{i=1} ^{N} K_i +\mcY$.

$(1. \Rightarrow  3.)$: Since the market is collectively $\mcY$-collectively complete, once any $f\in L^{0}(\Omega, \mathbf{F}_T,P)$ is fixed we have that both $\pm f$ are $\mcY$-collectively replicable. Then from Proposition \ref{diversi} Item 2 we have $\rho^{ \mcY}_+(f)=\rho^{ \mcY}_-(f)$.

$(3. \Rightarrow  1.)$: As $\rho^{ \mcY}_+(f)=\rho^{ \mcY}_-(f)$ for every $ f\in L^{0}(\Omega, \mathbf{F}_T,P)$, we have by Proposition \ref{diversi} Item 1 that both $\pm f$ are $\mcY$-collectively replicable, and so is $f$ in particular.

$(2. \Rightarrow 4.)$: We have $\{\mathbf{Q} \}= \Me$ by assumption. As $2. \Rightarrow 1.$, we can apply Lemma \ref{lemma:sticazzi} so that $\mathbf{NCA}(\mcY)$ and $\mathbf{NCA}(-\mcY)$ imply $\mathbf{NCA}(\mathrm{span}(\mcY))$. From Theorem \ref{IFTAP:cone}, $\mathcal{M}_e(\mathrm{span}(Y))\neq \varnothing$. Moreover $\mathcal{M}_e(\mathrm{span}(Y))\subseteq \Me=\{\mathbf{Q}\}$ so that  $\mathcal{M}_e(\mathrm{span}(Y))=\{\mathbf{Q}\}$. As a consequence
\[\sum_{i=1}^NE_{Q^i}[Y^i]=0 \quad \forall \;Y\in \mcY.\]
Now, for any $ f\in L^{0}(\Omega, \mathbf{F}_T,P)$ we have: $f=x+k+Y$ with $x \in \mathbb R^N$, $k \in {\sf X}_{i=1} ^{N} K_i$ and $Y \in \mcY$. Then $f^i=(x^i+E_{Q^i}[Y^i])+k^i +(Y^i - E_{Q^i}[Y^i])$. Since $\sum_{i=1}^NE_{Q^i}[Y^i]=0$ and $\R^N_0\subseteq \mcY$ then $Y+(-E_{\mathbf{Q}}[Y]) \in \widehat{\mcY}$, hence $f \in \R^N+ {\sf X}_{i=1} ^{N} K_i+ \widehat{\mcY}$, so that $f$ is $\widehat{\mcY}$-collectively replicable.

$ (4. \Leftrightarrow 5. \Leftrightarrow 6.)$: follows from  $ (1. \Leftrightarrow 2. \Leftrightarrow 3.)$ substituting $\mcY$ with $\widehat{\mcY}+\R^N_0$ and  noticing that: 
\begin{itemize}
\item $\widehat{\mcY}+\R^N_0$ satisfies Assumption \ref{ass:cone}, by Lemma \ref{lemma:hatyfinitegen}.
    \item  Since $(\widehat{\mcY}+\R^N_0) \subseteq \mcY$, then \textbf{NCA}$({\mcY})$ implies  
\textbf{NCA}$(\widehat{\mcY}+\R^N_0)$. Similarly,  \textbf{NCA}$(-{\mcY})$ implies  
\textbf{NCA}$(-(\widehat{\mcY}+\R^N_0))$

    \item $\mathcal{M}_e(\widehat{\mcY}+\R^N_0)=\mathcal{M}_e(\widehat{\mcY})$ (by direct verification), $\rho^{ \widehat{\mcY}}_+=\rho^{ \widehat{\mcY}+\R^N_0}_+$ (by \eqref{rhorho22} ) and $\rho^{ \widehat{\mcY}}_-=\rho^{ \widehat{\mcY}+\R^N_0}_-$ (by \eqref{rem:rhopm})
    \item $f$ is $\widehat{\mcY}$-replicable if and only if it is $(\widehat{\mcY}+\R^N_0)$-replicable, by \eqref{eqalityplusrn0}.
\end{itemize}

$(5. \Rightarrow 2.)$: It follows from  $\varnothing\neq \mathcal{M}_e(\mcY) \subseteq \mathcal{M}_e(\widehat{\mcY}) $.

From  $\mathcal{M}_e(\mcY) \subseteq \mathcal{M}_e(\widehat{\mcY}) $, we see that these sets must coincide when they are singletons. The remaining statements are consequences of $\mathcal{M}_e(\mcY) = \mathcal{M}_e(\widehat{\mcY}) $, the representation \eqref{pricing:hedging} and \eqref {pricing:hedgingBis} together with \eqref{rem:rhopm} and Lemma \ref{lemma:sticazzi}.

\end{proof}

\section{Examples} \label{example1}

In this section, we provide a collection of examples that serve both to illustrate the preceding results and to disprove some natural, albeit incorrect, conjectures suggesting that properties valid when admissible exchanges form a vector space extend verbatim to the case of finitely generated cones.
By construction, the superhedging functionals $\rho^{{\mcY}}_+$,\, $\rho^{\widehat{\mcY}}_+$,\, $\rho^N_+$
satisfy $\rho^{{N}}_+\geq \rho^{\widehat{\mcY}}_+\geq \rho^{{\mcY}}_+$.
While equality may occur in specific cases, as in collectively complete markets where $\rho^{ \widehat{\mcY}}_+=\rho^{ \mcY}_+$, strict inequalities are possible in general. 

\begin{example}
This example illustrates a situation in which $\rho^{\widehat{\mcY}}_+=\rho^{N}_+$.
As in Example \ref{ex1}, we consider again Example 7.2.3 of \cite{BDFFM25}, where the restriction of the set $\Me$ to $\mathcal F_1$ is given by
\[
\left\{
\left(
\left(\frac12 q,\,1-q,\,\frac12 q\right),
\left(\frac12 q,\,1-q,\,\frac12 q\right)
\right)
\,\middle|\,
0<q<1
\right\}.
\]
A straightforward computation yields
\[
\widehat{\mcY}
=
\bigl\{
(y,0,-y),\,(-y,0,y)
\;\big|\;
y\in\mathbb R
\bigr\}.
\]
Consequently,
$
\mathcal{M}_e\bigl(\widehat{\mcY}+\mathbb{R}^N_0\bigr)
=
\mathcal{M}_e(\widehat{\mcY})
=
M_e^1\times M_e^2,
$
and therefore, by \eqref{pricing:hedgingBis} and \eqref{rhoNN},
$
\rho^{\widehat{\mcY}}_+
=
\rho^{N}_+$.
This conclusion is not surprising. Indeed, in this example every component $\widehat Y^i$ of an exchange
$
\widehat Y\in\widehat{\mcY}
$
is individually replicable at zero initial cost by agent $i$, that is,
$
\widehat Y^i\in K_i,
\, i=1,2.
$
Hence, enlarging the set of admissible exchanges from $\{0\}$ to $\widehat{\mcY}$ does not create any additional hedging opportunities beyond those already available through the agents' individual markets.
\end{example}

\begin{example}\label{example22}
We provide an explicit example where $\mcY$ is a vector space such that $\R^N_0 \subseteq \mcY \subseteq \mcY_0$, $\mathrm{dim}(\mcY)$ {is finite}, \textbf{NCA}$(\mcY)$ holds true and the market is $\mcY$-collectively incomplete ($ \abs{\Me} =\infty) $, there exists a global arbitrage ($M^1_e \cap M^2_e = \emptyset$),  dim$(\widehat{\mcY})=2$ and $|\mathcal{M}_e(\widehat{\mcY})|=\infty$. 
Moreover, for an appropriate choice of the vector of claims $f$, we have 
$$\rho^{{\mcY}}_+(f)<\rho^{\widehat{\mcY}}_+(f)<\rho^N_+(f).$$

We consider a two-period market model, $\mathcal T=\{0,1,2\}$,  with two agents and two stocks and where each agent $i$, $i=1,2$, may invest only in the stock $X^i$ and in the riskless asset $X^0_t=1$ for all $t\in \mathcal T$. 
The evolution of the price processes is described in Figure \ref{figtree2}. We take $|\Omega|=8$, a common filtration: $\mathbb F=\mathbb F^1 =\mathbb F^2$, $\mcF_0=\{\emptyset, \Omega\},\mcF_2=\mathcal{P}(\Omega) $ and $$ \mcF_1=\sigma(A_1=\{\omega_1,\omega_2\},A_2=\{\omega_3,\omega_4\},A_3=\{\omega_5,\omega_6\}, A_4=\{\omega_7,\omega_8\}
).$$
\begin{figure}[hbt!]
\begin{center}
\tikzstyle{level 1}=[level distance=2cm, sibling distance=2.5cm,->]
\tikzstyle{level 2}=[level distance=1.5cm, sibling distance=1cm,->]

\tikzstyle{bag} = [text width=1.5em, text centered]
\tikzstyle{end} = []

\begin{tikzpicture}[grow=right, sloped]
\node[bag](c1){$8$}
    child {
        node[bag]{$4$}        
            child {
                node[end, label=right:
                    {$2$}](y18) {}
                edge from parent
                node[above] {}
                node[below]  {$\blue{1/2}$}
            }
            child {
                node[end, label=right:
                    {$6$}](y17) {}
                edge from parent
                node[above] {$\blue{1/2}$}
                node[below]  {}
            }
            edge from parent 
            node[above] {}
            node[below]  {$\blue{(1/2)-(q+q')}$}
    }
    child {
        node[bag]{$4$}        
            child {
                node[end, label=right:
                    {$3$}](y16) {}
                edge from parent
                node[above] {}
                node[below]  {$\blue{1/2}$}
            }
            child {
                node[end, label=right:
                    {$5$}](y15) {}
                edge from parent
                node[above] {$\blue{1/2}$}
                node[below]  {}
            }
            edge from parent 
            node[above] {$\blue{q'}$}
            node[below]  {}
    }
    child {
        node[bag]{$4$}        
            child {
                node[end, label=right:
                    {$2$}](y14) {}
                edge from parent
                node[above] {}
                node[below]  {$\blue{1/2}$}
            }
            child {
                node[end, label=right:
                    {$6$}](y13) {}
                edge from parent
                node[above] {$\blue{1/2}$}
                node[below]  {}
            }
            edge from parent 
            node[above] {$\blue{q}$}
            node[below]  {}
    }
    child {
        node[bag] {$12$}        
        child {
                node[end, label=right:
                    {$8$}] (y12){}
                edge from parent
                node[above] {}
                node[below]  {$\blue{3/4}$}
            }
            child {
                node[end, label=right:
                    {$24$}](y11) {}
                edge from parent
                node[above] {$\blue{1/4}$}
                node[below]  {}
            }
        edge from parent         
            node[above] {$\blue{1/2}$}
            node[below]  {}
    };

\node[bag](y21) at ([xshift=1.1cm]y11) {$\red{\omega_1}$};
\node[bag](y22) at ([xshift=1.1cm]y12) {$\red{\omega_2}$};
\node[bag](y23) at ([xshift=1.1cm]y13) {$\red{\omega_3}$};
\node[bag](y24) at ([xshift=1.1cm]y14) {$\red{\omega_4}$};
\node[bag](y25) at ([xshift=1.1cm]y15) {$\red{\omega_5}$};
\node[bag](y26) at ([xshift=1.1cm]y16) {$\red{\omega_6}$};
\node[bag](y27) at ([xshift=1.1cm]y17) {$\red{\omega_7}$};
\node[bag](y28) at ([xshift=1.1cm]y18) {$\red{\omega_8}$};
% %%%%%Secondo:
\node[bag](c2) at ([xshift=6cm]c1){$20$}
    child {
        node[bag]{$20$}        
            child {
                node[end, label=right:
                    {$16$}](z18) {}
                edge from parent
                node[above] {}
                node[below]  {$\blue{1/2}$}
            }
            child {
                node[end, label=right:
                    {$24$}](z17) {}
                edge from parent
                node[above] {$\blue{1/2}$}
                node[below]  {}
            }
            edge from parent 
            node[above] {}
            node[below]  {$\blue{1-((5/4)p'+(3/4)p)}$}
    }
    child {
        node[bag]{$4$}        
            child {
                node[end, label=right:
                    {$2$}](z16) {}
                edge from parent
                node[above] {}
                node[below]  {$\blue{1/2}$}
            }
            child {
                node[end, label=right:
                    {$6$}](z15) {}
                edge from parent
                node[above] {$\blue{1/2}$}
                node[below]  {}
            }
            edge from parent 
            node[above] {$\blue{(1/4)(p'-p)}$}
            node[below]  {}
    }
    child {
        node[bag]{$16$}        
            child {
                node[end, label=right:
                    {$12$}](z14) {}
                edge from parent
                node[above] {}
                node[below]  {$\blue{1/2}$}
            }
            child {
                node[end, label=right:
                    {$20$}](z13) {}
                edge from parent
                node[above] {$\blue{1/2}$}
                node[below]  {}
            }
            edge from parent 
            node[above] {$\blue{p}$}
            node[below]  {}
    }
    child {
        node[bag] {$24$}        
        child {
                node[end, label=right:
                    {$48$}] (z12){}
                edge from parent
                node[above] {}
                node[below]  {$\blue{1/4}$}
            }
            child {
                node[end, label=right:
                    {$16$}](z11) {}
                edge from parent
                node[above] {$\blue{3/4}$}
                node[below]  {}
            }
        edge from parent         
            node[above] {$\blue{p'}$}
            node[below]  {}
    };

\node[bag](z21) at ([xshift=1.1cm]z11) {$\red{\omega_1}$};
\node[bag](z22) at ([xshift=1.1cm]z12) {$\red{\omega_2}$};
\node[bag](z23) at ([xshift=1.1cm]z13) {$\red{\omega_3}$};
\node[bag](z24) at ([xshift=1.1cm]z14) {$\red{\omega_4}$};
\node[bag](z25) at ([xshift=1.1cm]z15) {$\red{\omega_5}$};
\node[bag](z26) at ([xshift=1.1cm]z16) {$\red{\omega_6}$};
\node[bag](z27) at ([xshift=1.1cm]z17) {$\red{\omega_7}$};
\node[bag](z28) at ([xshift=1.1cm]z18) {$\red{\omega_8}$};

\end{tikzpicture}
\end{center}
\caption{Tree for the stocks $(X^1,X^2)$ at times $t=0,1,2$.}
\label{figtree2}
\end{figure}
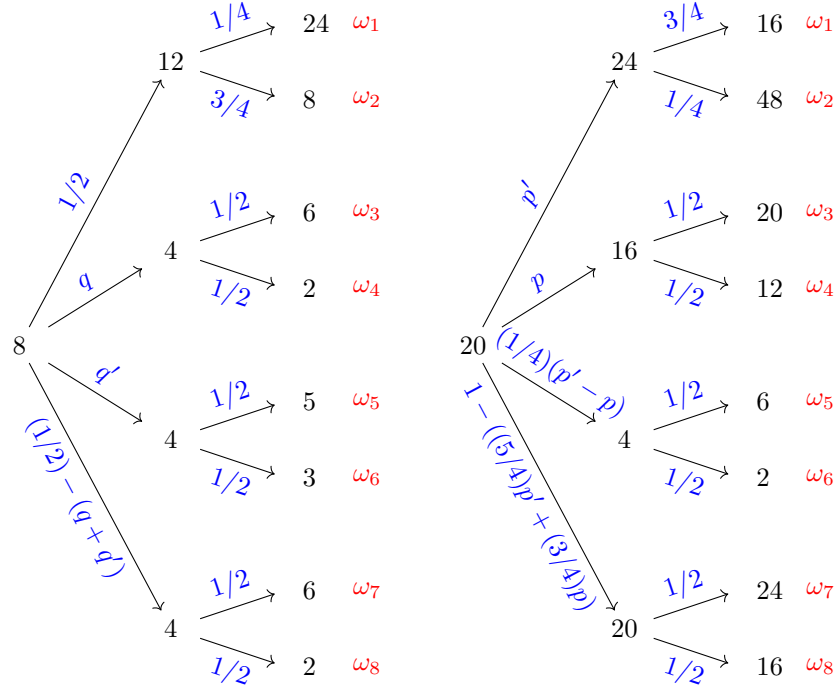

In Figure \ref{figtree2} the parameters $q,q',p,p'$ parametrizing the sets of equivalent martingale measures for the two stocks satisfy:
\begin{equation}
\label{eqconstraint}
    0<q<\frac12,\,0<q'<\frac12-q\quad\text{and}\quad 0<p<\frac12,\,p<p'<\frac45-\frac35p
\end{equation}

In the following, to simplify the computations, we will pick $f \in  L^0(\Omega, \mathbf{F}_1, P)$ with $\mcF_1$-measurable components, so that we only care about restrictions of measures to $\mcF_1$. We express measures $Q$ on $\mcF_1$ as $4$-tuples, in the form $(Q(A_1),Q(A_2),Q(A_3),Q(A_4))$.

Denoting by $M^i_e|_{\mcF_1}$ the collection of restrictions to $\mcF_1$ of elements in $M^i_e$, one readily verifies  that
\begin{equation}\label{mesonf1}
    \begin{split}
     M^1_e|_{\mcF_1}&=\bigg\{\bigg(\frac12, q,q', \frac12-(q+q')\bigg),0<q<\frac12,\,0<q'<\frac12-q \bigg\},\\
    M^2_e|_{\mcF_1}&=\bigg\{\bigg(p', p,\frac14(p'-p), 1-\bigg(\frac54 p'+\frac34 p\bigg)\bigg), 0<p<\frac12,\,p<p'<\frac45-\frac35p\bigg\}.   
    \end{split}
\end{equation}
We also select 
$$\mathcal{Y} = \left\{ Y \in (L^0(\Omega, \mathcal{F}_1, P))^2 \mid Y^1+Y^2 = 0  \right\}.$$
%$\mcY=\mcY_0\cap L^0(\Omega,\mcF_1,P)$. 
With this choice of $\mcY$, the set of collective equivalent  martingale measures becomes
 \begin{align*}
\mathcal M_{e}(\mcY) &:=\bigg\{ \mathbf Q=(Q^1,Q^2)  \mid Q^i \in  {M}^i_{e} \, \text{ and }   \,  E_{Q^1}[Y^1]+E_{Q^2}[Y^2] = 0 \text {   }\forall \,Y \in \mcY \bigg\}\\
&=\bigg\{ \mathbf Q=(Q^1,Q^2)  \mid Q^i \in  {M}^i_e \, \text{ s.t. } Q^1=Q^2 \text{ on } \mathcal F_1 \bigg\}
\end{align*}
and thus the set of restrictions to $\mcF_1$ of elements in $\Me $  is
\begin{equation}\label{MeMe}
\Me |_{\mcF_1}=\bigg\{\bigg(\bigg(\frac12, q,\frac18-\frac14 q, \frac38-\frac34q\bigg),\bigg(\frac12, q,\frac18-\frac14 q, \frac38-\frac34q
\bigg)\bigg), 0<q<\frac12\bigg\}.
\end{equation}
From  $\widehat{\mcY}:=\left \{ Y \in  \mcY \mid E_{Q^1}[Y^1]=E_{Q^2}[Y^2]=0 \,\, \text{ for all } \, (Q^1,Q^2) \in \Me \right \} \subseteq \mcY$,
we can now identify 
\begin{align*}
    \widehat{\mcY}=\big\{\big(&-(w+3z)1_{A_1}+(w+3z)1_{A_2}+4w1_{A_3}+4z1_{A_4},\\&(w+3z)1_{A_1}-(w+3z)1_{A_2}-4w1_{A_3}-4z1_{A_4} ), z,w\in \R \big\},
\end{align*}
that is $\widehat{\mcY}=\mathrm{span}(\widehat{Y}_1,\widehat{Y}_2)$
for $\widehat{Y}_1=(\widehat{Y}_1^1,\widehat{Y}_1^2)$, $\widehat{Y}_2=(\widehat{Y}_2^1,\widehat{Y}_2^2)$, where
\begin{align*}
\widehat{Y}_1^1&:=-1_{A_1}+1_{A_2}+41_{A_3}\, \text{ and }\, \widehat{Y}_1^2:=-\widehat{Y}_1^1 \\ 
    \widehat{Y}_2^1&:=-31_{A_1}+31_{A_2}+41_{A_4}\, \text{ and }\, \widehat{Y}_2^2:=-\widehat{Y}_2^1.
\end{align*}

Hence,
$\mathcal{M}_e(\widehat{\mcY})|_{\mcF_1}$ consists of all the pairs $(Q_1,Q_2)$ such that $Q_i\in M^i_e|_{\mcF_1}$ and such that
$E_{Q_1}[\widehat{Y}^1_1]+E_{Q_2}[\widehat{Y}^2_1]=0=E_{Q_1}[\widehat{Y}^1_2]+E_{Q_2}[\widehat{Y}^2_2]$ (note that $\widehat{Y}_1,\widehat{Y}_2$ are both $\mcF_1$-measurable).
Using \eqref{mesonf1} and these conditions we get the values $q'=\frac18-\frac14q$ and $p'=\frac12$. 

Imposing that $p<p'<4/5-3/5p$ we get that $q$ and $p$ satisfy: 
$0<q<1/2,0<p<1/2$.
We conclude that 
$$\mathcal{M}_e(\widehat{\mcY})|_{\mcF_1}=\bigg\{\bigg(\bigg(\frac12, q,\frac18-\frac14q, \frac38-\frac34q\bigg),
       \bigg(\frac12, p,\frac18-\frac14p,\frac38-\frac34p\Big)\bigg), 0<p,q<\frac12\bigg\}.$$

     Observe that the only element $\widehat{Y}=(\widehat{Y}^1,\widehat{Y}^2) \in \widehat{\mcY}$ having each component $\widehat{Y}^i$ replicable by investing only in the market $(X^0,X^i)$ is the zero element $(\widehat{Y}^1,\widehat{Y}^2)=(0,0)$.

Now pick $f^1=1_{A_3}, f^2=-f^1$ so that  $f=(f^1,f^2)\in\mcY$ but $(f^1,f^2)\notin \widehat{\mcY}$. By \eqref{pricing:hedging} and \eqref{pricing:hedgingBis} we have 
       \begin{align*}
           \rho^{{\mcY}}_+(f)&=0, \\ \rho^{\widehat{\mcY}}_+(f)&=\sup_{q,p\in(0,1/2)}\bigg(\frac18-\frac14q-\Big(\frac18-\frac14p\Big)\bigg)=\frac18,  \\
           \rho^N_+(f)&= \sup_{q\in (0,1/2)}\sup_{q'\in (0, 1/2-q)}q'+\sup_{p\in (0,1/2)}\sup_{p'\in (p, 4/5-3/5p)}-\frac14(p'-p)=\frac12,
       \end{align*}
       where the first equality follows from the fact that $f\in\mcY$, the second line of equalities follow picking $q\downarrow 0,p\uparrow\frac12$, the third one choosing $q\downarrow 0,q'\uparrow\frac12,p\uparrow \frac12,p'\downarrow\frac12$.
Hence for this choice of $f$ we have $\rho^{{\mcY}}_+(f)<\rho^{\widehat{\mcY}}_+(f)<\rho^N_+(f)$ as claimed.\\

\end{example}

\begin{example}\label{example42}
In this example $N=2$, $\mcY \subseteq \mcY_0$ is a convex cone satisfying Assumption \ref{ass:cone}, \textbf{NCA}$(\mcY)$ holds true but \textbf{NCA}$(-\mcY)$ fails, the market is $\mcY$-collectively complete, namely $\mathbb R^2+K_1\times K_2+\mathcal Y=(L^0(\Omega,\mathcal F,P))^2$, however $|\Me|=\infty$ and, for some $f \in (L^0(\Omega,\mathcal F,P))^2$ such that $ f $ and $(-f)$    are ${\mcY}$-collectively replicable, it holds:
\[
\rho_-^{\mathcal Y}(f)<\rho_+^{\mathcal Y}(f).
\]
In contrast with Proposition \ref{diversi} Item 2, this example proves that when \textbf{NCA}$(-\mcY)$ fails, the equality between superhedging and subhedging of a $\mcY$-replicable and $(-\mcY)$-replicable claim $f$ does not necessarily hold.

Let \(N=2\), $T=1$,
$\Omega=\{\omega_1,\omega_2,\omega_3\}$
with 
$\ P(\{\omega_j\})>0,\quad j=1,2,3.$
We take $\mathcal{F}=\mathcal{F}^i_1=\mathcal{P}(\Omega)$ and identify \(L^0(\Omega,\mathcal{F},P)\) with \(\mathbb R^3\). 
We consider trivial initial $\sigma$ algebras.  In addition to the investment in the riskless asset $X^0_t=1$, $t=0,1$, agent $1$ has access to the stock $X^1$ whose increments are given by
\[
\Delta X^1=(-1,0,1):=a,
\]
and agent $2$ has access to the stocks $X^2$ and $X^3$ whose increments are given by

\[
\Delta X^2=(-2,1,0):=b_1,
\qquad 
\Delta X^3=(-3,0,1):=b_2.
\]
We set
\[
K_1:=\operatorname{span}(\{a\}),\qquad
K_2:=\operatorname{span}(\{b_1,b_2\})
\]
and we recognize that both markets are arbitrage free, as the sets of equivalent martingale measures
for \(K_1\), and respectively for \(K_2\), are
\[
M^1_e=\{Q^1_t=(t,1-2t,t),\quad 0<t<\frac12\},
\]
\[
M^2_e=\left \{\bar Q^2:=\left(\frac16,\frac13,\frac12\right)\right\}.
\]
While the market of agent $1$ is incomplete, the market of agent $2$ is complete and thus 
\[
K_2=\{z\in\mathbb R^3:E_{\bar Q_2}[z]=0\}, \quad \mathbb R\mathbf 1+K_2=\mathbb R^3.
\]
Let
\[
V_1:=(1,-2,1),\qquad V_2:=(-4,2,2)
\]
and notice that $V_2=-V_1+3a$.
Let $\mathbf 1=(1,1,1)$ and define the cone of admissible exchanges by
\[
\mathcal Y
:=
\operatorname{cone}
\big\{
(\mathbf 1,-\mathbf 1),
(-\mathbf 1,\mathbf 1),
(V_1,-V_1),
(V_2,-V_2)
\big\} \subseteq \mcY_0.
\]
Equivalently, every \(Y\in\mathcal Y\) is of the form
$Y=(y,-y)$,
where
\[
y=c\mathbf 1+\alpha V_1+\beta V_2,
\quad c\in\mathbb R,\quad \alpha,\beta\ge 0.
\]
In particular, every admissible exchange is componentwise zero-sum:
$Y^1+Y^2=0$.

To determine $\mathcal M_e(\mathcal Y)$, compute
$E_{Q^1_t}[V_1]
=
6t-2, \, E_{\bar Q^2}[V_1]
=
0$.
Then the condition induced by \((V_1,-V_1)\in\mathcal Y\) is
\[
E_{Q^1_t}[V_1]-E_{\bar Q^2}[V_1]\le 0
\quad\Longleftrightarrow\quad
6t-2\le 0
\quad\Longleftrightarrow\quad
t\le \frac13.
\]
Similarly,
$E_{Q^1_t}[V_2]
=
2-6t, \, E_{\bar Q^2}[V_2]
=
1 
$
and thus the condition induced by \((V_2,-V_2)\in\mathcal Y\) is
\[
E_{Q^1_t}[V_2]-E_{\bar Q^2}[V_2]\le 0
\quad\Longleftrightarrow\quad
1-6t\le 0
\quad\Longleftrightarrow\quad
t\ge \frac16.
\]
Therefore
\[
\mathcal M_e(\mathcal Y)
=
\left\{
\left(Q^1_t,\bar Q^2\right)\mid
\frac16\le t\le \frac13
\right\}.
\]
In particular, \(\mathcal M_e(\mathcal Y)\neq\varnothing\) and hence, by Theorem \ref{IFTAP:cone}, \textbf{NCA}$(\mcY)$ holds true.
By reversing the above inequalities we also conclude that  \(\mathcal M_e(-\mathcal Y)= \varnothing\) and thus, by the same theorem, \textbf{NCA}$(-\mcY)$ fails.

We next prove that the market is \(\mathcal Y\)-collectively complete, namely
\[
\mathbb R^2+K_1\times K_2+\mathcal Y
=
(L^0(\Omega,\mathcal F,P))^2.
\]
Let \((g_1,g_2)\in (L^0(\Omega,\mathcal F,P))^2 \simeq \mathbb R^3\times\mathbb R^3\) be arbitrary.
The vectors \(\mathbf 1,a,V_1\) form a basis of \(\mathbb R^3\), since
\[
\det
\begin{pmatrix}
1&-1&1\\
1&0&-2\\
1&1&1
\end{pmatrix}
=
6\neq 0.
\]
Therefore there exist \(r,\eta,\gamma\in\mathbb R\) such that
$g_1=r\mathbf 1+\eta a+\gamma V_1$.
Write
$\gamma=\gamma^+-\gamma^-,\,
\gamma^+,\gamma^-\ge 0$,
and define
\[
y:=\gamma^+V_1+\gamma^-V_2. 
\]
Then $(y,-y) \in \mcY$ and, since \(V_2=-V_1+3a\), we get
\[
y
=
(\gamma^+-\gamma^-)V_1+3\gamma^-a
=
\gamma V_1+3\gamma^-a.
\]
Hence
\[
g_1
=
r\mathbf 1+(\eta-3\gamma^-)a+y=x_1\mathbf 1+k_1+y
\]
for some \(x_1\in\mathbb R\), \(k_1\in K_1\), and some \(y\) such that \((y,-y)\in\mathcal Y\).

Now consider the second component and  recall that
$\mathbb R\mathbf 1+K_2=\mathbb R^3$.
Therefore, applying this to \((g_2+y) \in \R^3\), there exist \(x_2\in\mathbb R\) and \(k_2\in K_2\) such that
\[
g_2+y=x_2\mathbf 1+k_2, \,  \text{ namely } \, g_2=x_2\mathbf 1+k_2-y.
\]
Thus
\[
(g_1,g_2)
=
(x_1,x_2)+(k_1,k_2)+(y,-y) \in \mathbb R^2+K_1\times K_2+\mathcal Y.
\]
Since $(g_1,g_2) \in (L^0(\Omega,\mathcal F,P))^2$ was arbitrary,
the market is \(\mathcal Y\)-collectively complete.\\
Now take any $f \in (L^0(\Omega,\mathcal F,P))^2$. Then $\pm f \in (L^0(\Omega,\mathcal F,P))^2 =\mathbb R^2+K_1\times K_2+\mathcal Y$, so that both $ f $ and $(-f)$    are ${\mcY}$-collectively replicable.

We now show that 
 the collective superhedging price and subhedging price of $f:=(V_1,-V_1)$ are not equal.
For every
$\mathbf Q=(Q^1_t,\bar Q^2)\in\mathcal M_e(\mathcal Y)$,
we have
\[
E_{Q^1_t}[f^1]+E_{\bar Q^2}[f^2]
=
E_{Q^1_t}[V_1]-E_{\bar Q^2}[V_1]
=
6t-2
\]
and from 
$\frac16\le t\le \frac13$
we get
$-1\le 6t-2\le 0$.
Therefore, by the collective pricing--hedging duality,
\[
\rho_+^{\mathcal Y}(f)
=
\sup_{\mathbf Q\in\mathcal M_e(\mathcal Y)}
\left(
E_{Q^1}[f^1]+E_{Q^2}[f^2]
\right)
=
0,
\]
and
\[
\rho_-^{\mathcal Y}(f)
=
\inf_{\mathbf Q\in\mathcal M_e(\mathcal Y)}
\left(
E_{Q^1}[f^1]+E_{Q^2}[f^2]
\right)
=
-1.
\]
Finally we observe that, by Theorem \ref{BaseFTAPII},
\[
\Pi(f)
=
\{(6t-2,0)\mid1/6\le t\le 1/3\}
=
[-1,0]\times\{0\},
\]
which is not relatively open.
\end{example}

\begin{example}[\(\mathbf{NCA}(\mathcal Y)+\mathbf{NCA}(-\mathcal Y)
\not\Rightarrow
\mathbf{NCA}(\operatorname{span}(\mathcal Y))\)]
\label{example43}
Let
\[
\Omega=\{\omega_1,\omega_2,\omega_3\},
\qquad 
\mathbb P(\{\omega_j\})>0,\quad j=1,2,3,
\]
and take \(N=2\), $T=1$. We take $\mathcal{F}=\mathcal{F}^i_1=\mathcal{P}(\Omega)$.
We identify \(L^0(\Omega,\mathcal{F})\) with \(\mathbb R^3\). 
We consider trivial initial $\sigma$ algebras.  Agent $i$ has access to the stock $X^i$ whose increments are given by
\[
\Delta X^1:=(2,-1,0),
\qquad 
\Delta X^2:=(-1,2,0).
\]
Then  $K_1=\mathrm{span}(\{\Delta X^1\}),
\, 
K_2:=\mathrm{span}(\{\Delta X^2\})$.
Let
$\mathbf 1=(1,1,1),
\, 
e_j:=\mathbf 1_{\{\omega_j\}},\, j=1,2,3
$
and consider the following cone of admissible exchanges 
\[
\mathcal Y
:=
\operatorname{cone}
\Big\{
(\mathbf 1,-\mathbf 1),
(-\mathbf 1,\mathbf 1),
(e_1,-e_1),
(e_2,-e_2)
\Big\}.
\]
 We will now show that  $\mathcal M_e(\mathcal Y)\neq \emptyset$, $\mathcal M_e(-\mathcal Y)\neq \emptyset$ but $\mathcal M_e({\mathrm{span}}(\mathcal Y))= \emptyset$. Applying Theorem \ref{IFTAP:cone} we see that both $\mathbf{NCA}(\mathcal Y)$ and $\mathbf{NCA}(-\mathcal Y)$ hold, but $
\mathbf{NCA}(\mathrm{span}(\mathcal Y))$ does \textbf{not} hold.
Let
\[
Q^1=(q^1_1,q^1_2,q^1_3),
\qquad
Q^2=(q^2_1,q^2_2,q^2_3),
\]
with all coordinates strictly positive and summing to \(1\). The individual martingale
conditions are
$
E_{Q^1}[\Delta X^1]=0,
E_{Q^2}[\Delta X^2]=0$. 
These conditions become
$2q^1_1-q^1_2=0,
-q^2_1+2q^2_2=0$.
Therefore
\[
M^1_e=\bigg\{Q^1_s=(s,2s,1-3s),
\,0<s<\frac13\bigg\},\qquad M^2_e=\bigg\{Q^2_t=(2t,t,1-3t),
\,0<t<\frac13\bigg\}
\]
We first compute \(\mathcal M_e(\mathcal Y)\). By definition,
$(Q^1_s,Q^2_t)\in M_e(\mathcal Y)$
if and only if
\[
E_{Q^1_s}[Y^1]+E_{Q^2_t}[Y^2]\le 0
\qquad \forall\,Y\in\mathcal Y.
\]
Since every \(Y\in\mathcal Y\) is of the form $Y=(Y^1,-Y^1)$, this is equivalent to
\[
E_{Q^1_s}[Y^1]-E_{Q^2_t}[Y^1]\le 0 \quad \forall\, Y^1 \text { such that } (Y^1,-Y^1) \in\mathcal Y.
\]
The terms involving \(\mathbf 1\) impose no restriction, since \(Q^1_s\) and \(Q^2_t\)
are probability measures. Hence it is enough to impose the inequalities for \(e_1\)
and \(e_2\).
For \(e_1\), we obtain
$
E_{Q^1_s}[e_1]-E_{Q^2_t}[e_1]
=
s-2t
\le 0
$,
and hence
$
s\le 2t.
$
For \(e_2\), we obtain
$
E_{Q^1_s}[e_2]-E_{Q^2_t}[e_2]
=
2s-t
\le 0
$,
and hence
$
2s\le t
$.
Therefore,
\[
\Me=\left\{
\bigl(Q^1_s,Q^2_t\bigr) \mid
Q^1_s=(s,2s,1-3s),\
Q^2_t=(2t,t,1-3t),\
0<s<\frac13,\ 0<t<\frac13,\, 2s\le t
\right\} \neq \emptyset.
\]
%A simple choice satisfying these constraints is
%$
%s=\frac{1}{10},
%t=\frac{1}{4}$, and we see that $M_e(\mathcal Y)\neq \emptyset$.
To compute \(\mathcal{M}_e(-\mathcal Y)\) we just need to reverse the additional inequalities above. 
Consequently
\[
\mathcal{M}_e(-\mathcal Y)=\left\{
\bigl(Q^1_s,Q^2_t\bigr)\mid
Q^1_s=(s,2s,1-3s),\
Q^2_t=(2t,t,1-3t),\
0<s<\frac13,\ 0<t<\frac13,\, s \geq 2t
\right\} \neq \emptyset.
\]
%A choice satisfying these constraints is
%$
%s=\frac14,
%t=\frac{1}{10}$, proving $M_e(-\mathcal Y)\neq \emptyset$.
Finally, to compute \(\mathcal M_e(\mathrm{span}(\mathcal Y))\), the additional inequalities above must become equalities, which implies  
$
s=2t$ and $
2s=t$
so that
$s=t=0$, which is impossible because equivalent probability measures must
assign strictly positive mass to every state. Hence
$
\mathcal M_e(\mathrm{span}(\mathcal Y))=\emptyset.
$
\end{example}

\paragraph{Acknowledgements}
The authors thank Frank Riedel for pointing out the remarkable relation between our framework and the one in \cite{Page2006,DanaLeVan2010} and subsequent works.

\paragraph{Data availability statement} Data sharing is not applicable to this article as no datasets were generated or analyzed during the current study.

\paragraph{Funding} A. Doldi and M. Frittelli are members of GNAMPA-INDAM.

{
\bibliographystyle{abbrv}  
\bibliography{BibAll}
}

%%%%%%%%%%%%%%%%%%%%%%%%%%%%%%%%%%

\end{document}